\documentclass[sigconf]{aamas} 

\usepackage{balance} 

\setcopyright{ifaamas}
\acmConference[AAMAS '24]{Proc.\@ of the 23rd International Conference
on Autonomous Agents and Multiagent Systems (AAMAS 2024)}{May 6 -- 10, 2024}
{Auckland, New Zealand}{N.~Alechina, V.~Dignum, M.~Dastani, J.S.~Sichman (eds.)}
\copyrightyear{2024}
\acmYear{2024}
\acmDOI{}
\acmPrice{}
\acmISBN{}

\acmSubmissionID{813}

\title[Distributed Online Rollout for Multivehicle Routing in Unmapped Environments]{Distributed Online Rollout for Multivehicle Routing in Unmapped Environments}

\author{Jamison W Weber}
\affiliation{
  \institution{Arizona State University}
  \city{Tempe, Arizona}
  \country{United States}}
\email{jwweber@asu.edu}

\author{Dhanush R Giriyan}
\affiliation{
  \institution{Arizona State University}
  \city{Tempe, Arizona}
  \country{United States}}
\email{dgiriyan@asu.edu}

\author{Devendra R Parkar}
\affiliation{
  \institution{Arizona State University}
  \city{Tempe, Arizona}
  \country{United States}}
\email{dparkar1@asu.edu}

\author{Dimitri P Bertsekas}
\affiliation{
  \institution{Arizona State University}
  \city{Tempe, Arizona}
  \country{United States}}
\email{dbertsek@asu.edu}

\author{Andr\'{e}a W Richa}
\affiliation{
  \institution{Arizona State University}
  \city{Tempe, Arizona}
  \country{United States}}
\email{aricha@asu.edu}

\begin{abstract}
We consider a generalization of the well-known multivehicle routing problem: given a network, a set of agents occupying a subset of its nodes, and a set of tasks, we seek a minimum cost sequence of movements subject to the constraint that each task is visited by some agent at least once. The classical version of this problem assumes a central computational server that observes the entire state of the system perfectly and directs individual agents according to a centralized control scheme. In contrast, we assume that there is \emph{no centralized server} and that each agent is an individual processor with \emph{no a priori knowledge of the underlying network, including the locations of tasks and other agents.} Moreover, our agents possess \emph{strictly local communication and sensing capabilities} (restricted to a fixed radius), aligning more closely with several real-world multiagent applications. We present a \emph{fully distributed, online, and scalable reinforcement learning algorithm} for this problem whereby agents self-organize into local clusters to which they independently apply a multiagent rollout scheme. We demonstrate empirically via extensive simulations that there exists a critical sensing radius beyond which the distributed rollout algorithm begins to improve over a greedy base policy. This critical sensing radius grows proportionally to the $\log^*$ function of the size of the network and is therefore a small constant in practice. Our decentralized reinforcement learning algorithm achieves approximately a factor of two cost improvement over the base policy for a range of radii between two and three times the critical sensing radius. In addition, we observe in simulations that our distributed algorithm requires \emph{exponentially fewer computational resources} than the centralized approach, at only a \emph{small constant factor detriment in cost}. We validate our algorithm through physical robot experiments in continuous space and show that the resulting behavior reflects the discrete simulations.
\end{abstract}


\newcommand{\BibTeX}{\rm B\kern-.05em{\sc i\kern-.025em b}\kern-.08em\TeX}

\usepackage[utf8]{inputenc} 
\usepackage[T1]{fontenc}    
\usepackage{hyperref}       
\usepackage{url}            
\usepackage{booktabs}       
\usepackage{amsfonts}       
\usepackage{nicefrac}       
\usepackage{microtype}      
\usepackage{xcolor}         
\usepackage{hyperref}

\usepackage{amsmath,amsthm}
\usepackage{algorithm}
\usepackage[noend]{algpseudocode}
\usepackage{graphicx}
\usepackage{hyperref}
\usepackage{mathtools}
\usepackage{xcolor}
\usepackage{caption}
\usepackage{subcaption}
\usepackage{siunitx}
\graphicspath{ {./images/} }
\newtheorem{theorem}{Theorem}[section]
\newtheorem{lemma}[theorem]{Lemma}

\makeatletter
\newcommand{\algmargin}{\the\ALG@thistlm}
\makeatother
\newlength{\whilewidth}
\newlength\myindent
\newif\ifcomment
\commenttrue    

\newcommand{\junk}[1]{}

\algdef{SE}[parWHILE]{parWhile}{EndparWhile}[1]
  {\parbox[t]{\dimexpr\linewidth-\algmargin}{%
     \hangindent\whilewidth\strut\algorithmicwhile\ #1\ \algorithmicdo\strut}}{\algorithmicend\ \algorithmicwhile}%
\algnewcommand{\parState}[1]{\State%
  \parbox[t]{\dimexpr\linewidth-\algmargin}{\strut #1\strut}}

\makeatletter
\gdef\@copyrightpermission{
	\begin{minipage}{0.3\columnwidth}
		\href{https://creativecommons.org/licenses/by/4.0/}{\includegraphics[width=0.90\textwidth]{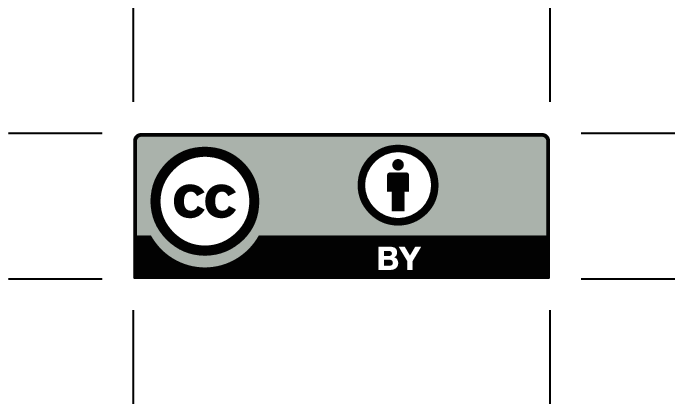}}
	\end{minipage}\hfill
	\begin{minipage}{0.7\columnwidth}
		\href{https://creativecommons.org/licenses/by/4.0/}{This work is licensed under a Creative Commons Attribution International 4.0 License.}
	\end{minipage}
	\vspace{5pt}
}
\makeatother

\begin{document}


\pagestyle{fancy}
\fancyhead{}


\maketitle 

\section{Introduction}
In this work, we consider a generalization of the well-studied multivehicle routing problem (MVRP).
The MVRP involves a set of agents (vehicles) and a set of tasks each occupying a position on an underlying network.
The objective is to compute a sequence of movement decisions for each vehicle such that the total number of vehicle movements is minimized, subject to the constraint that each task is visited at least once by some agent.
Classically, these sequences are computed by a centralized controller with complete state information that 
may freely direct all agents according to the control sequences it computes.

%
Our generalization addresses scenarios where the {\em network topology, including the initial positions of agents and task locations, is not known a priori}, and where a {\em centralized controller is not available}, as is potentially the case in many real-world scenarios, such as minefield disarmament \cite{minefield}, post-disaster
search and rescue \cite{golden_chapter_2014}, and unmanned aerial vehicle navigation \cite{uav}. 
In such settings, 
each agent must be in turn capable of performing individual computations and of {\em local sensing and inter-agent communication},
both limited to a radius around the agent's location in the network (e.g. as in packet radio networks).
Since no centralized computation is available, local communication between agents is necessary for coordination and information exchange. Since sensing is limited, exploration is required to locate the tasks before completing them.  We refer to this as the {\em Unmapped Multivehicle Routing Problem with Local Constraints (UMVRP-L)}. 
Note that MVRP can be viewed as a special case of UMVRP-L where the sensing radius is infinite.

The UMVRP-L is challenging for agent coordination algorithm design. 
Algorithmically, one may consider a reinforcement learning solution, as has been done with other MVRP variants~\cite{raza}, but as the underlying network is unknown, methods that require offline computation such as traditional Q-learning~\cite{sutton2018reinforcement}, policy approximation (via the training of parametric architectures)~\cite{neuro}, or executing policy iteration to near convergence~\cite{10.5555/1671238} may not be applicable.
Moreover, even online Q-learning requires extensive exploration over a single trajectory to learn Q-values, and also frequently revisits states to ensure convergence~\cite{onlineQ}. 
This exploration factors into the total cost of solving the problem.

This realization suggests rollout~\cite{bertsekasMAR} as a natural candidate, as it is a reinforcement learning algorithm that (under certain conditions) admits a fully online solution without requiring precomputation or iterative convergence, and produces reliably favorable results due to its fundamental connections to Newton's method~\cite{bertsekasMAR,newton,Bertsekas2021LessonsFA}.
However, rollout is a centralized control scheme that typically assumes information is universally shared among agents and a central dispatch (see \emph{classical information pattern}~\cite{bertsekasMAR}), and hence we present a decentralized version of rollout to address UMVRP-L.
In general, the rollout algorithm requires a base policy, which provides a benchmark for comparison.
For our distributed modification of rollout specifically, this comparison allows us to evaluate the impact of
\emph{communication} and \emph{coordination} in our distributed framework:
Communication refers to agents sharing sensed information, whereas coordination refers to agents organizing in such a way that they collectively solve a problem intelligently.  

\subsection{Our Contributions}
We present a distributed (decentralized) version of the 
Multiagent Rollout algorithm from~\cite{bertsekasMAR}, which we call Decentralized Multiagent Rollout (DMAR).
In this approach, we assume each agent runs DMAR independently
and in parallel with other agents.
DMAR uses simple local rules to induce iterative self-organization of the agent set into temporary 
clusters of bounded size.
These clusters generate control sequences locally using multiagent rollout (equipped with a greedy nearest-neighbor base heuristic~\cite{bertsekasMAR}) and execute them before systematically dissolving. Agents alternate between self-organizing into clusters and randomly searching the environment for tasks not visible to any cluster, making DMAR applicable to {\em fully unmapped environments}, unlike the general multiagent optimization literature.


We formally prove 
that DMAR is \textit{probabilistically complete}~\cite{probabilisticCompleteness}, i.e. for any solvable instance, the probability that {\em DMAR correctly solves UMVRP-L} converges to unity as the running time tends to infinity.
We also show that each \emph{round} of DMAR terminates in constant time and that the \emph{expected number of rounds required to solve the instance is bounded by $\mathcal{O}(N^2)$}, where $N$ is the number of nodes in the network. 
Moreover, we conduct extensive 
simulations that compare the costs incurred by DMAR with those of a modification of DMAR that we call \emph{the greedy-exploration policy}---a base policy that strictly uses a greedy nearest-neighbor base heuristic (without application of multiagent rollout) within the clusters generated. 
Note that the greedy-exploration policy
does not involve coordination of the control policies within each cluster.
In these experiments, we modulate the sensing radius over a wide variety of network topologies and sizes, the number of tasks, and agent-to-task ratios.

We show 
that there is a {\em critical sensing radius} 
beyond which DMAR begins to improve over the greedy-exploration policy: The critical sensing radius 
is closely approximated by
$\log_2^*(N)$,\footnote{$\log_2^*(N)=1+\log_2^*(\log_2(N))$, for $N>1$ and $\log_2^*(N)=0$, for $N\leq 1$.
   Note that, e.g,  $\log_2^*(N)\le 5$ for any $N\le 10^{10000}$.} and is therefore a small constant in practice regardless of network size.
%
We show empirically that for sensing radii between 
$2\log_2^*(N)$ and $3\log_2^*(N)$ (we call those the \emph{effective radii}), 
DMAR generates trajectories of roughly a {\em half
the cost of those produced by the base policy} on average. 
We compare the costs of DMAR to that of the centralized multiagent rollout algotithm~\cite{bertsekasMAR}.  For the effective radii range, 
DMAR shows an {\em exponential improvement in computational costs over 
the centralized multiagent rollout algorithm}, while incurring just {\em a constant factor (approximately 3) increase in movement costs}.


Lastly, we validate our algorithm through physical robot experiments in continuous space using the Robotarium platform~\cite{robotarium} as a proof-of-concept.
The results reflect what we observed in discrete-space simulations, demonstrating the algorithm's robustness to sensor noise, which would underscore the potential for actual deployments of the algorithm 
in scenarios where it is not reasonable to assume that a complete mapping of tasks and agents exists.


\subsection{MVRP and Other Related Work}
A discussion on reinforcement learning techniques already appeared in the introduction, so here we focus more on the MVRP literature and agent-based distributed systems. As  the MVRP is NP-hard~\cite{vehicle_routing} (even when the network is planar~\cite{planar}), several heuristics have been proposed (e.g.,~\cite{bertsekasMAR, liu2023heuristics}).
Researchers have considered many fully observable and centralized variants that introduce further constraints and complexities such as limited vehicle capacity, heterogeneous capacities, route length limits, time windows when
tasks may be completed, multiple depots, multi-criteria objective functions, 
stochastic time and transition variants, and demand constrained and weighted MVRP, e.g.~\cite{toth,Han_2018,BRAEKERS2016300}.

Most relevant to UMVRP-L is probably the \textit{dynamic MVRP} variant~\cite{toth, psaraftis}, where certain input information such as the network, vehicle availability, or task locations may not be known in advance, but  becomes available over time (perhaps probabilistically).
Dynamic MVRP assumes a single centralized dispatcher with instantaneous access to information as it becomes available, instant communication with all vehicles, and persistent access to all previously available information---all of which we do not assume in this work.
Although UMVRP-L has commonalities with dynamic MVRP, it is fundamentally different due to its decentralized nature and strict local constraints. To the best of our knowledge, variants of vehicle routing where a centralized planner is not available have not yet been considered by researchers, and hence, no benchmark algorithms for UMVRP-L exist in the literature.

AI researchers have also considered multiagent problems where the full state is not directly observable to the agents, but
the proposed solutions in this domain require centralized computation---often executed offline~(e.g. \cite{bengio,bhattacharya,foerster,gupta,kar}).
A formulation known as a \emph{decentralized partially observable Markov decision process} (Dec-POMDP)  was introduced in~\cite{decpomdp} that combines notions from distributed computing and optimal control in partially observable environments, but optimally solving Dec-POMDPs was shown to be intractable. 
In these formulations, the environment is partially observable to each agent, but the information sensed by each agent is not universally shared.
Approximate solutions to Dec-POMDPs were explored in~\cite{macrocontrols,planning_decpomdp, policy_decpomdp}, but ultimately relied on intensive offline and centralized simulations.

More broadly, there have been several recent advances in the field of distributed computing on engineering collective behavior for systems of agents.
Such agents are typically restricted to only local computations and interactions with other agents to achieve a desired system-wide behavior, as in population protocols~\cite{population_protocols} and programmable matter~(e.g. \cite{amoebot,chirikjian,self-assembly,hexagon,modular}).

\section{Model and preliminaries}
\label{sec:model}

In general, one assumes a regular discretization of 2D space for the MVRP, such as a 2-dimensional \emph{rectangular grid graph} (as in~\cite{10.5555/1671238}), as we do in this work. An instance of the MVRP consists of a rectangular grid graph $G(V,E)$ (e.g. in Figure~\ref{fig:soac-example}) on $N$ nodes, a set $S$ of $m$ \emph{agents}, where each agent occupies a node in the graph (co-location is permitted), and 
a subset $\tau$ of nodes that 
represent \emph{task} locations in $G$.
We also assume there is a subset $O$ of nodes that agents may not occupy, which allows us to represent \emph{obstacles} in the environment (represented by the black squares in Figure~\ref{fig:soac-example}).
Moreover, each agent has a unique ID.
This model operates in discrete time steps, during which an agent may move to a node adjacent to its current location.
Whenever an agent visits a task node, that task is considered \emph{completed}, and is removed from the network.
At each time step, a unit cost is incurred for each agent that moves from one node to an adjacent node. A control sequence $Y_i$ for agent $i$ is given by the sequence of nodes $v_{(i,1)}, v_{(i,2)}, \ldots$ visited by the agent.
A solution $Y$ is given by $m$ sequences $Y_1,\ldots,Y_m$ of controls corresponding to each agent in $S$.
The cost $C(Y)$ of a solution $Y$ is the sum of the $m$ control sequence lengths, i.e. $C(Y)=\sum_{Y_i\in Y}|Y_i|$.
A solution $Y$ is optimal whenever $C(Y)$ is minimum.

The UMVRP-L is a generalization of MVRP where $G,O,S,\tau,$ are not known a priori and no centralized computation is permitted.
We assume that each agent has some on-board computational capability and enforce that any computation must be executed by some agent, as opposed to a central controller.
At any given time, each agent $a$ occupies a position $\rho(a)\in V$.
We also define
the \emph{view} of an agent $a$, denoted as $\nu(a)$, to be the subgraph of the grid induced by $\rho(a)$ and \emph{all} nodes (including obstacles) within at most $k$ hops of $\rho(a)$, where $k$ is an integral radius parameter.
An agent's view describes what it is able to sense from its location (i.e. obstacles, 
neighboring nodes, agents, and tasks). 
Depending on the application context, there may be other (more appropriate) choices for the definition of a view, such as those where obstacles restrict what an agent can observe.
We provide a discussion of the impact of various view definitions in Appendix-A~\cite{weber2023distributed}.
Let $\nu(S')=\bigcup_{a\in S'}\nu(a)$, and refer to this as the \emph{view} of an agent set $S'\subseteq S$.



Our distributed model operates under \emph{parallel synchronous 
time steps}.
That is, at each discrete time step $t$, all agents are \emph{activated} simultaneously in parallel, and thus may only act upon information communicated at time $t-1$.
An {\em agent activation} may consist of either moving, performing some 
computation, or simply waiting.
For an agent $a$, the input of a computation performed by $a$ is limited to 
$\nu(a)$ (including the memories of the agents therein), and any agent $b\in \nu(a)$ (including $a$) may have their memory modified directly by $a$ as an output to that computation.

\paragraph{\bf\em Rollout and Multiagent Rollout} 
From a centralized perspective, one can view an instance of MVRP as an infinite horizon dynamic programming problem that can be solved approximately in an online setting using \emph{approximation in value space}, of which the rollout algorithm is a special case~\cite{bertsekasMAR}.
That is, we approximate the \emph{cost-to-go} terms of the Bellman equations via online heuristics. 
Denote $\Omega$ as the state space representing all possible agent and task configurations over $G$, plus an absorbing terminal state representing the completion of all tasks.
Let $U(x)$ denote the set of $m$-dimensional control vectors permissible at state $x\in \Omega$.
A system dynamics function $f(x,u)$ returns the resulting state obtained from application of control $u\in U(x)$ at state $x\in \Omega$.
Lastly, denote $g(x,u)$ as the cost function indicating the total number of vehicle movements at state $x$ given control $u$.
Rollout generates an approximately optimal state and control sequence $x_0,\tilde{u}_0,x_1,\tilde{u}_1,\ldots$
where each successive control $\tilde{u}_s$  at state $x_s$ is determined by
\begin{equation}
\tilde{u}_s\in \arg\min_{u\in U(x_s)}\{g(x_s,u)+\tilde{H}(f(x_s,u))\}
     \label{eqn:bellman}
\end{equation}
and results in arrival at state $x_{s+1}$.
Here $\tilde{H}(f(x_s,u))$ refers to the approximate cost-to-go obtained from simulation of a heuristic (also known as a \textit{base policy}).

Note that the minimization operation in Equation~\ref{eqn:bellman} is exponential in $m$, however, a recent reformulation of a multiagent problem to an equivalent problem with a larger state space, but simpler control space was  described by Bertsekas in~\cite{bertsekasMAR}. 
This reformulation was the basis for substantial simplification of rollout and policy iteration algorithms for multiagent problems. 
The resulting algorithm (called {\it multiagent rollout}) dramatically reduces the computational complexity of each minimization operation of the standard rollout algorithm (from exponential to linear in the number of agents), while maintaining its cost improvement properties. 
This exponential gain arises from a sequential rollout control computation on an agent-by-agent basis, where the control for each agent is computed based on a partial list of controls computed for each agent earlier in the sequence, hence greatly restricting the number of multi-component controls considered.
Specifically, the control $u^a$ for each agent $a$ is computed in a sequence $u^{a_1},u^{a_2},\ldots,u^{a_m}$ such that when $u^{a_i}$ is computed, each $u^{a_j}$ for $1\le j<i$ has been determined, and for any $u^{a_j}$ where $i<j\le m$, and for all future controls, a base policy is used to estimate the cost-to-go in Equation~\ref{eqn:bellman}.
See~\cite{bertsekasMAR} for a comprehensive description.

\section{Decentralized Multiagent Rollout}
\label{sec:algorithms}
In this section, we present and analyze our Decentralized Multiagent Rollout Algorithm (DMAR), described as a series of phases
that, when executed in sequence and repeated, will solve the UMVRP-L.
DMAR is a fully distributed and local algorithm that
applies multiagent rollout from~\cite{bertsekasMAR} 
locally, and therefore can handle unmapped environments. 
In addition to their (unique) IDs, DMAR assumes that the agents know a constant parameter $\psi$ (common to all agents), which will determine an upper bound on the size of the clusters formed by DMAR. During its execution, the algorithm will build trees of agents, each of which corresponds to a current cluster, and so each agent also maintains a parent pointer, a constant number $c$ of child pointers, and various flags with values \textsc{T,F}  used for message passing. 
At a high level,
a {\em round of DMAR} consists of the following sequence of {\em phases}, which the agents collectively repeat:
\begin{enumerate}
\item \textbf{Self-Organizing Agent Clusters (SOAC)}: Agents self-or\-ganize into clusters of constant size $c^{\mathcal{O}(\psi)}$.
For each cluster $K$, the agents in $K$ will locally form a tree $\mathcal{T}_K$, 
whose root will become the cluster leader.
If two agents share an edge in $\mathcal{T}_K$, then they appear in their respective views.

\item\textbf{Local Map Aggregation (LMA)}: For each cluster $K$, agents in $K$
route their views through $\mathcal{T}_K$ to the cluster leader.

\item \textbf{Team-Restricted Multiagent Rollout (TMAR)/Execute Movement (EM)}: For each cluster $K$, the leader computes a set $\mathcal{R}$ of control sequences (one for each agent in $K$) using multiagent rollout 
on the assembled view information. $\mathcal{R}$ is then broadcast over $\mathcal{T}_K$.
If an agent does not belong to a cluster, it follows a random trajectory of bounded length. Else, each agent moves along its assigned control trajectory, after which it unassigns itself from its cluster. 
\end{enumerate}

\begin{algorithm}
\caption{Local DMAR round protocol for agent $a$}
\label{alg:overview}
\begin{algorithmic}[1]
    \State{Set $a.msgFlag\gets a.ldrFlag\gets\textsc{F}$; Set $a.clusterID\gets a.joinMessage\gets a.children\gets a.parent\gets\textsc{null}$.}
    \State{\textbf{if} there is a task in $\nu(a)$, set $a.ldrFlag\gets\textsc{T}$.}\Comment{begin SOAC} \label{alg:soac-start}
    \If {$a.ldrFlag$ and $\exists x\in \nu(a):x.ldrFlag$ and $x.ID>a.ID$}
        \State{set $a.ldrFlag\gets\textsc{F}$.}
    \EndIf
    \State{\textbf{if} $a.ldrFlag$ \textbf{then} set $a.clusterID\gets a.ID$.}  
    
    \Repeat  \Comment{form initial clusters}\label{alg:formclusters}
        \If{ $a.clusterID=\textsc{null}$}
        \If {there is an agent $x\in\nu(a)$ with $x.clusterID\ne\textsc{null}$ \\ \indent\indent\indent and $|x.children|<c$} \label{start-append}
        \State{Append $a$ to $x.children$; set $a.parent\gets x$.}
        \State{Set $a.clusterID\gets x.clusterID$.}\label{end-append}

        \Else 
        \If {there are agents $ x_1,x_2,\ldots ,x_r \in$  $\nu(a)$ from distinct \indent\indent clusters, where $r\le c$} \Comment{cluster-join}\label{alg:clusterjoin-start}
            \State{Set $a.ClusterID\gets a.ID$.} 
            \For{each $x_i$}
            \State{\textbf{try} set $x_i.joinMessage\gets a.ID$.  }\label{try}
            \State{}\Comment{Note that $x_i$ may reject the join request.} 
            \EndFor
        \EndIf
          \EndIf  
            \EndIf
            \Repeat
                \If{$a.joinMessage=x.ID$}
                    \State{Set $a.ClusterID\gets x.ClusterID$.}
                    \For{each child $b$ of $a$}
                    \State{\textbf{if} $x\ne b$, \textbf{try} set $b.joinMessage\gets a.ID$.}
                    \EndFor
                    \If{$a.parent\ne x$}
                    \State{\textbf{try} set $a.parent.joinMessage\gets a.ID$.}
                    \EndIf
                    \State{Set $a.parent\gets x$; append $a$ to $x.children$.}
                    \State{Set $a.children\gets a.joinMessage\gets\textsc{null}$.}
                \EndIf
            \Until{$\mathcal{L}(\psi)$ (synchronized) iterations} \label{clusterjoin-end}

    \Until{$\lceil\log_2{\psi}\rceil$ (synchronized) iterations}  \label{alg:soac-end}
    \If{$a.clusterID\ne \textsc{null}$} \Comment{begin LMA} \label{alg:lma-start}
        \State{Construct map $\mathcal{M}=\nu(a)$.}
        \State{\textbf{If} $a.parent.clusterID=a.parent.ID$, set $a.msgFlag\gets\textsc{T}$.}
        \Repeat \label{downloop}
            \If{$a.msgFlag$}
                \State{Set $\mathcal{M}(a) \gets \mathcal{M}(a)\cup \mathcal{M}(a.parent)$, preserving}
                \Statex{\indent\indent\indent coordinates relative to $\mathcal{M}(a.parent)$.}
                \State{\textbf{if} $a.children\ne\textsc{null}$, set $a.msgFlag\gets\textsc{F}$.}
                \State{\textbf{for each} child $b$, set $b.msgFlag\gets\textsc{T}$.}
               
            \EndIf
        \Until{$\mathcal{L}(\psi)$ (synchronized) iterations}
        \Repeat \label{uploop}
            \If{$a.msgFlag$ and $a.clusterID\ne a.ID$}
                \State{Set $\mathcal{M}(a.parent) \gets \mathcal{M}(a)\cup \mathcal{M}(a.parent)$.}
                \State{Set $a.parent.msgFlag\gets\textsc{T}$; $a.msgFlag\gets\textsc{F}$.}
            \EndIf
        \Until{$\mathcal{L}(\psi)$ (synchronized) iterations}   \label{alg:lma-end}
    \EndIf

    
        \If{$a$ is the leader of a cluster $K$} \Comment{begin TMAR/EM} \label{alg:TMAR-start}
            \State{Compute $
            \mathcal{R}\gets \textsc{MR}(\mathcal{M}(K))$ and send $\mathcal{R}$ to each child.}\label{alg:rollout}
       \EndIf     
        \Repeat
            \State{\textbf{if} $a.clusterID\ne \textsc{null}$ and $a$ has received $\mathcal{R}$ from parent, \indent \textbf{then} send $\mathcal{R}$ to each child of $a$.}        
     \Until{$\mathcal{L}(\psi)$ (synchronized) iterations}
     \State{\textbf{if} $a.clusterID\ne\textsc{null}$ \textbf{then} move $a$ according to $\mathcal{R}^{a.ID}$.} \label{alg:EM-start}

        \State{\textbf{else} $a$ executes at most $\lambda(\psi)$ controls, each selected uniformly at random until there is a task in $\nu(a)$.}\label{alg:EM-end}
    

     
\end{algorithmic}
\end{algorithm}

We will show that each \emph{round} 
of DMAR terminates in constant time, and that the number of rounds required for DMAR to solve an instance of UMVRP-L is polynomially bounded in expectation (Theorem~\ref{thm:DMARexpected}).
A pseudocode description of a round of DMAR is given in Algorithm~\ref{alg:overview}.
We denote $\mathcal{L}(\psi)=\frac{3}{2}(\psi-2)$ 
as the maximum possible height of any agent tree (see Lemma~\ref{lemma:SOAC2}) and reference the function $\textsc{MR}(\mathcal{M})$, which takes an instance $\mathcal{M}$ of MVRP and returns a set of control sequences $\mathcal{R}$ as computed by multiagent rollout (see Section~\ref{sec:model}).
Note that each agent runs DMAR individually on their respective processor, but all agents remain synchronized in their execution of DMAR. 
We achieve this by ensuring that each phase is executed using exactly the same number (a function of $\psi$) of operations regardless of which agent is performing the computation.
For example, an agent that is not in a cluster will wait approximately $2\cdot\mathcal{L}(\psi)$ time steps for Lines~\ref{alg:lma-start}-\ref{alg:lma-end} before continuing. 

\paragraph{\textbf{Self-Organizing Agent Clusters} (Algorithm~\ref{alg:overview}, Lines~\ref{alg:soac-start}-\ref{alg:soac-end})}
To achieve the clustering, agents that see tasks become temporary cluster \emph{leaders} via a local leader election scheme.
The clusters then grow by iteratively appending nearby agents (limited by parameter $\psi$).
At the start of SOAC, each agent sets a leader flag ($ldrFlag$) to \textsc{true (T)} if it sees a task.
Then, if an agent $a$ sees a task while simultaneously seeing other agents that each see tasks, $a$ disqualifies itself from becoming a leader if any such agent has a larger ID value than $a$'s by setting its leader flag to \textsc{false (F)}. 
Any remaining agents after the above process form leaders of new unique clusters.
Next, SOAC iteratively appends agents to existing clusters over multiple 
time steps.
At iteration $i$ of the loop beginning in Line~\ref{alg:formclusters}, if there exists an agent $a$ that does not belong to a cluster, $a$ searches its view for any agent $a'$ that already belongs to some cluster $K'$ and joins $K'$ by becoming a child of $a'$. 
Consequently, as each cluster $K$ grows, a bi-directional (through parent and child pointers) tree structure $\mathcal{T}_K$ is maintained that manages the flow of messages in $K$.
Moreover, when $a$ joins a cluster $K$, it obtains $K$'s cluster ID  from its parent in $\mathcal{T}_K$, which corresponds to the ID of the leader of $K$.

We found in preliminary simulations that communication between adjacent clusters generally leads to significantly better solutions.
As such, for adjacent clusters, SOAC employs a mechanism to combine them into potentially larger clusters called \emph{cluster-join} (Lines~\ref{alg:clusterjoin-start}-\ref{clusterjoin-end}).
For any agent $a$ that is not in a cluster at iteration $i$ after the appending process  (Lines~\ref{start-append}-\ref{end-append}), if $a$ sees other agents in distinct clusters, $a$ becomes the leader of a new (possibly larger) cluster that attempts to include all agents in neighboring clusters.
Agent $a$ sends messages ($joinMessage$) to all such agents in its view, which are then broadcast throughout the respective agent trees.
These messages request that its recipients adopt the new cluster ID (that of agent $a$) and reconfigure their tree pointer structure such that their tree is rooted at $a$.
If an agent receives multiple cluster reassignment messages simultaneously, it will disregard all but one (chosen arbitrarily, see line~\ref{try}).
Note that because of the possibility that an agent may receive multiple cluster reassignment requests simultaneously, not all agents of these trees may join the new cluster rooted at $a$, since some may join other clusters in this manner.
However, in simulations we observed this generally resulted in larger clusters, which led to greatly improved performance. 
Each agent continues propagating join messages for $\mathcal{L}(\psi)$ 
time steps per iteration.
After a SOAC run terminates, the process yields a collection of clusters (and possibly some agents that belong to no clusters).
Generally, if the would-be parent of an agent $a$ attempting to join a cluster has no available child pointers, then $a$ simply does not join the cluster.
Figure~\ref{fig:soac-example} provides a visualization of a SOAC execution.
Lemma~\ref{lemma:SOAC2} guarantees properties of the tree structure of each cluster formed by SOAC, and bounds the running time.


\begin{figure}
     \centering
     \begin{subfigure}[b]{.9\textwidth/4}
         \centering
         \includegraphics[width=\textwidth]{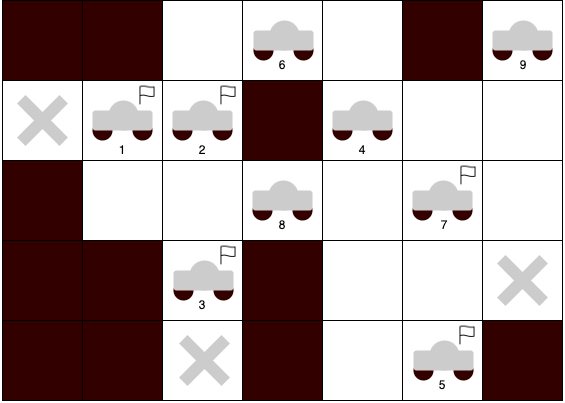}
        
         \label{fig:soaca}
     \end{subfigure}
     \hfill
     \begin{subfigure}[b]{.9\textwidth/4}
         \centering
         \includegraphics[width=\textwidth]{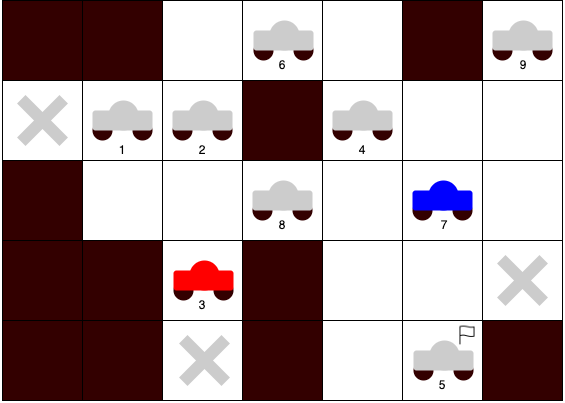}
         
         \label{fig:soacb}
     \end{subfigure}
    \par\bigskip
     \begin{subfigure}[b]{.9\textwidth/4}
         \centering
         \includegraphics[width=\textwidth]{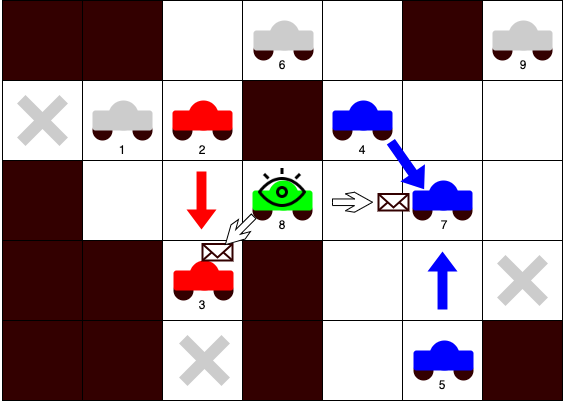}
         
         \label{fig:soacc}
     \end{subfigure}
     \hfill
     \begin{subfigure}[b]{.9\textwidth/4}
         \centering
         \includegraphics[width=\textwidth]{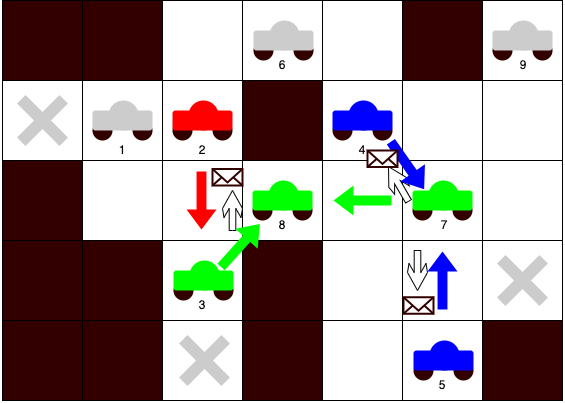}
         
         \label{fig:soacd}
     \end{subfigure}
     \par\bigskip
     \begin{subfigure}[b]{0.9\textwidth/4}
         \centering
         \includegraphics[width=\textwidth]{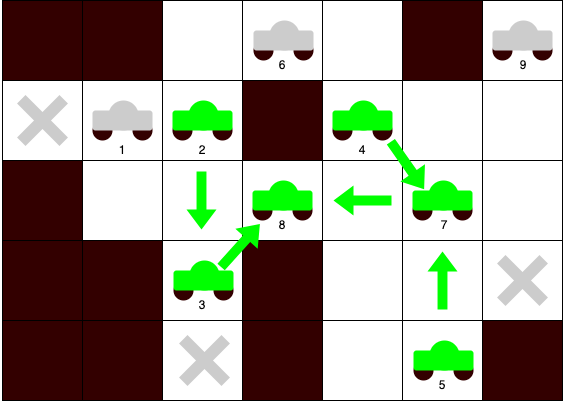}
         
         \label{fig:soace}
     \end{subfigure}
     \hfill
     \begin{subfigure}[b]{0.9\textwidth/4}
         \centering
         \includegraphics[width=\textwidth]{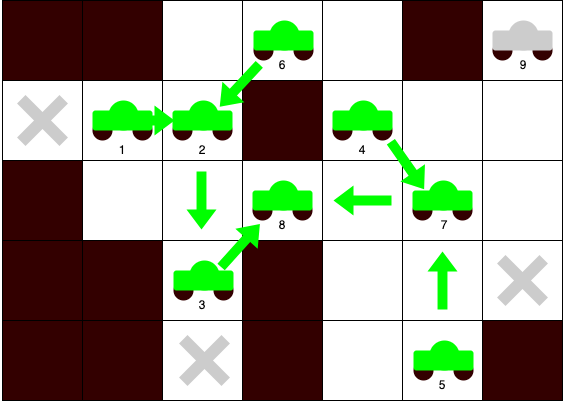}
         
         \label{fig:soacf}
     \end{subfigure}
        \caption{SOAC Execution. Read left to right, top to bottom: Cars represent agents. Gray crosses are tasks. Agent colors indicate cluster membership; gray agents are in no cluster. IDs are indicated below each agent. Flags indicate which agents see tasks. Assume $k=2$, $\psi=4$. (top left) No agent is in a cluster; 1,2,3,5,7 see tasks. (top right) 1,2,5 see agents with larger IDs who see tasks, so disqualify themselves. 3,7 become leaders of new clusters. (middle left) Iteration 1 begins. 2,4,5 join clusters, colored arrows indicate tree pointers. 8 sees agents in different clusters. 8 declares self new leader and initiates a join with messages. (middle right) 3 and 7 join 8's cluster. Join messages are propagated to 2,4,5  (bottom left) 2,4,5 join the green cluster. (bottom right) Iteration 2 begins. 1,6 join green cluster. 9 was too distant so joins no cluster.   }
        \label{fig:soac-example}
\end{figure}

\begin{lemma}
\label{lemma:SOAC2}
Let $K\subseteq S$ be an arbitrary cluster of agents in $G$ obtained at the end of a round of SOAC.
Then
(i) $K$ forms an agent tree $\mathcal{T}_K$ rooted at a leader $\ell$ with height at most $\mathcal{L}(\psi)=\mathcal{O}(\psi)$ and $|\mathcal{T}_K|= c^{\mathcal{O}(\mathcal{L}(\psi))}$ agents; (ii) all agents in $\mathcal{T}_K$ associate their membership in $K$ with $\ell$; and (iii) SOAC terminates in $\mathcal{O}(\psi\log_2{\psi})$ 
time steps.
\end{lemma}

\begin{proof}
    (i) Each tree is formed around a single leader.
    An unassigned agent appends itself as a child of a single parent node already in $\mathcal{T}_K$, for some current cluster $K$, adopting the cluster ID $K$ of its parent (which can be traced back to the leader of $K$). Thus $\mathcal{T}_K$ will form a single acyclic connected component, for each $K$.
    In the cluster-join process (initiated by an agent $b$), the parent pointer of an agent $a$ will always point to the agent from which it received (and accepted) a join request, which can likewise be traced back to $b$. 
    As agent $a$ removes its child pointers after propagating the request, it becomes a leaf of the new tree, and any of its former children may undergo the same process in the next loop iteration.
    At a loop iteration $i$, we claim that the height of the largest cluster tree first increases by at most one (Lines~\ref{start-append}-\ref{end-append}) and then at most doubles (plus one for the joining agent $b$) during the cluster-join process.
    This follows since for any agent $x$ in $\mathcal{T}_K$, the neighborhood of $x$ in $\mathcal{T}_K$ (i.e. not including $b$) after the cluster-join process is an induced subgraph of $x$'s previous neighborhood, and any neighbors it loses are no longer in $\mathcal{T}_K$.
    As $x$ does not gain new neighbors in $\mathcal{T}_K$, 
    the longest path in $\mathcal{T}_K$ is either the same or some shorter path (in the case where not all of $K$ remains in the tree) after pointer reassignment.
    The longest possible path in any undirected rooted tree $\mathcal{T}$ is 
    at most twice the height of $\mathcal{T}$.
    Thus, the maximum tree height at iteration $i$ is given by $\mathcal{L}(i)=0$ if $i=1$, and $\mathcal{L}(i)\le 2(\mathcal{L}(i-1)+1)+1$ for $i>1$. 
    The solution of this relation is $\frac{3}{2}(2^i-2)=\mathcal{O}(2^{\log_2{\psi}})=\mathcal{O}(\psi)$. 
    The maximum number of agents in $\mathcal{T}_K$ is $\sum_{j=1}^{\mathcal{L}(\psi)}c^j=\frac{c^{\mathcal{L}(\psi)+1}-1}{c-1}\le c^{\mathcal{O}(\mathcal{L}(\psi))}$.
    The height bound on an agent tree implies waiting $\mathcal{L}(\psi)$ 
    time steps per iteration is enough time for join messages to propagate through the tallest tree, and so all agents in $\mathcal{T}_K$ will agree on the leader ID, implying (ii). Counting shows that SOAC terminates in $\mathcal{O}(\psi\log{\psi})$ 
    time steps.
\end{proof}

\paragraph{\textbf{Local Map Aggregation} 
 (Algorithm~\ref{alg:overview}, Lines~\ref{alg:lma-start}-\ref{alg:lma-end})}
At the end of SOAC, a subset of the agents is organized into clusters, for each of which there exists a respective leader.
In the second phase, the objective is that for each independent cluster $K$ with leader $\ell$, agent $\ell$ obtains a map of $K$, which reflects the collective view of the cluster ($\nu(K)$).
Moreover, this map has a relative coordinate labelling scheme whereby the coordinates of $\ell$ are $(0,0)$.
At a high level, agents in $K$ convey their views to their neighbors in $\mathcal{T}_K$.
These views then aggregate until $\ell$ obtains a view of the entire cluster.

An agent $a$ possesses a \emph{map} structure $\mathcal{M}(a)$ that includes $\nu(a)$, as well as possibly the views of other agents (initially $\mathcal{M}(a)=\nu(a)$).
Note that their may be agents in the view of $a$ that do not belong to the cluster of which $a$ is a member, but $a$ disregards these.
A \emph{complete map} $\mathcal{M}(K)$ of a cluster $K$ includes the union of views of every agent in $K$, i.e. $\mathcal{M}(K)=\nu(K)$.
For two maps $\mathcal{M}$ and $\mathcal{M}'$, we use $\mathcal{M}\cup \mathcal{M}'$ to denote the component-wise union of the maps.
Agent $\ell$ will initiate the propagation of messages down to the leaves of $\mathcal{T}_K$.
Agent $\ell$ generates a map $\mathcal{M}(\ell)=\nu(\ell)$, which has nodes labelled with relative coordinates such that $\rho(\ell)=(0,0)$.
Agent $\ell$ passes this map to each child $a$, and agent $a$ sets $\mathcal{M}(a)\gets \mathcal{M}(a)\cup \mathcal{M}(\ell)$.
During this union operation, agent $a$ extends the relative coordinates of $\mathcal{M}(\ell)$ to $\mathcal{M}(a)\cup \mathcal{M}(\ell)$.
This message passing process from parent to child repeats using a message flag ($msgFlag$). 
After a sufficient amount of time passes, each leaf $a$ sends $\mathcal{M}(a)$ (which represents a union of maps obtained from the path from $\ell$ to $a$ in $\mathcal{T}_K$) to its respective parent $b$, and updates its map $\mathcal{M}(b)\gets\mathcal{M}(b)\cup \mathcal{M}(a)$ (again preserving relative coordinates).
This process then iterates $\mathcal{L}(\psi)$ times from child to parent where then $\ell$ receives a complete map.
Agents repeatedly propagate messages long enough for the height of the tallest tree to be traversed twice (down then up) to preserve a relative coordinate scheme that is common to all agents in $K$.
If $\mathcal{M}(K)$ contains tasks that are unreachable by some agent, $\ell$ disregards their respective components, and if this results in a map without tasks, the cluster is dissolved via broadcast of messages that cause agents to reset their cluster associations and pointers.
We omit this description from Algorithm~\ref{alg:overview} for brevity.
The lemma below guarantees correctness and running time.

\begin{lemma}
\label{lemma:LMA}
Given 
$\mathcal{T}_K$, local map aggregation provides the leader $\ell$ of the cluster $K$ with a complete map $\mathcal{M}(K)=\nu(K)$ with all nodes labelled relative to $\rho(\ell)=(0,0)$ in $\mathcal{O}(\psi)$ 
time steps.
\end{lemma}
\begin{proof}
    As the size of any cluster is constant, and an agent $a$ in cluster $K$ will disregard any agent not in $K$, $\mathcal{M}(a)$ has constant size and is generated in constant time.
    By Lemma~\ref{lemma:SOAC2}, the maximum height of $\mathcal{T}_K$ is $\mathcal{L}(\psi)$, and so $\mathcal{O}(\mathcal{L}(\psi))$ is a sufficient time so that the leaves of $\mathcal{T}_K$ receive partial cluster maps from $\ell$.
    Likewise, the broadcast of map information from the farthest leaf to the root requires $\mathcal{O}(\mathcal{L}(\psi))$ 
    time steps.
    As each message pass from an agent $a$ to an agent $b$ produces a union of $\mathcal{M}(a)$ and $\mathcal{M}(b)$, after $\mathcal{O}(\mathcal{L}(\psi))=\mathcal{O}(\psi)$ 
    time steps, $\ell$ possesses a complete map of $K$.
    Note that the number of messages passed from a parent to its children is constant since $c$ is constant.
    By the labelling scheme given in LMA, the complete map obtained by leader $\ell$ is labelled with coordinates relative to $\rho(\ell)=(0,0)$.
    This follows since the map originates from $\ell$, where it is labelled with relative coordinates to $\rho(\ell)=(0,0)$, and every map union operation appends nodes to a map that contains the original map of $\ell$, so the relative labelling scheme is preserved in the complete map. 
\end{proof}

\paragraph{\textbf{Team-Restricted Multiagent Rollout and Execute Movement} 
 (Algorithm~\ref{alg:overview}, Lines~\ref{alg:TMAR-start}-\ref{alg:EM-end})}

TMAR describes how to apply multiagent rollout from~\cite{bertsekasMAR} to an independent cluster $K$ while respecting all of the local sensing and communication restrictions of our model.
After the sequential execution of SOAC followed by LMA, each cluster $K$ is structured as a spanning tree $\mathcal{T}_K$ rooted at the leader agent $\ell$, who possesses a complete map $\mathcal{M}(K)$.
Agent $\ell$ executes multiagent rollout on $\mathcal{M}(K)$ with a greedy nearest-neighbor base policy (i.e. each agent moves toward its respective nearest task) to obtain control sequences for each agent in $K$, which terminates in constant time since the size of $\mathcal{M}(K)$ is bounded by a constant due to the constant height and branching factor of any agent tree.
Agent $\ell$ then initiates a broadcast of these sequences (via parent-to-child message passing in $\mathcal{T}_K$), each of which is associated with a specific agent ID. 
After enough iterations have passed for each agent $a$ to receive its corresponding control sequence, $a$ will simply execute its control sequence via movement (EM).
To maintain synchronization between all clusters, $a$ will wait an additional number of time steps until it reaches a maximum control sequence length $\lambda(\psi)=\mathcal{O}([\mathcal{L}(\psi)]^4)$---a function on $\psi$ that returns a bound on the maximum route length for a greedy nearest-neighbor heuristic applied to $\mathcal{M}(K)$.
For any agent $a$ not in a cluster, instead of following a computed control sequence, $a$ executes a (non-empty) sequence of controls generated uniformly at random of length at most $\lambda(\psi)$ and waits as soon as it sees a task, thereby exploring the network.

The theorem below shows that, in constant time, a round of DMAR separates a subset of the agents into clusters, for each of which control sequences are generated from multiagent rollout locally and executed.
Agents that are not members of clusters perform random walks of at most constant length in each run of EM.
\begin{theorem}
\label{thm:Cluster-MAHR}
During a round $r$ of DMAR, the agent set is partitioned into members of clusters and non-members. 
Within $\mathcal{O}(\psi)$ time steps, each agent $x$ that is in a cluster $K$ executes a control sequence $\mathcal{R}^x$ that is generated from multiagent rollout with a greedy base policy applied to $\mathcal{M}(K)$, and executes $\mathcal{R}^x$ within $\mathcal{O}(\psi^{4})$ 
time steps.
If $x$ is not in a cluster, $x$ performs a random walk of length at most $\lambda(\psi)=\mathcal{O}(\psi^{4})$.
Lastly, round $r$ terminates in constant time.
\end{theorem}
\begin{proof}
As $\mathcal{O}(\psi)$ additional time is required for message broadcast in the event $K$ is dissolved at the end of LMA (due to the tree height), by Lemmas~\ref{lemma:SOAC2} and~\ref{lemma:LMA}, the agent set is partitioned into agents that are members of clusters and those who are not in constant time.
 Those that are in a cluster $K$ are associated with an agent tree $\mathcal{T}_K$, the root of which possesses a map $\mathcal{M}(K)$ of $K$.
 By Lemma~\ref{lemma:SOAC2}(i) and since $k$ is constant, the radius of $\mathcal{M}(K)$ is bounded by $\mathcal{O}(\psi)$, and hence the area of $\mathcal{M}(K)$ is bounded by $\mathcal{O}(\psi^{2})=\mathcal{O}(1)$.
 The number of agents in $K$ is bounded by a constant $(c^{\mathcal{O}(\mathcal{L}(\psi))})$.
 As the area of $\mathcal{M}(K)$ is of constant size, it can contain at most a constant number of tasks (i.e. at most one per node).
 Hence, any multiagent rollout computation on a map $\mathcal{M}(K)$ will terminate in constant time.
 Consequently, by Lemma~\ref{lemma:SOAC2}(i) and according to TMAR, each agent in a cluster $K$ obtains a control sequence computed from multiagent rollout applied to $\mathcal{M}(K)$ in a constant number of 
 time steps and moves accordingly.

In the worst case, there may be a task at every node, and a route may visit every task with a cost incurred proportional to the diameter of the cluster.
This implies the longest route for a cluster $K$ is bounded by $\lambda(\psi)=\mathcal{O}(\psi^{4})=\mathcal{O}(1)$ since the diameter of the cluster is bounded by its area.
Since the base policy is greedy, multiagent rollout can perform no worse than it (see cost-improvement properties~\cite{bertsekasMAR}), and hence can never produce a trajectory longer than $\lambda(\psi)$.
As each agent spends a constant number of 
time steps in TMAR, and exactly $\lambda(\psi)$ 
time steps in EM, these phases terminate constant time.
If an agent $x$ is not in a cluster, then $x$ executes a non-empty random walk of length at most $\lambda(\psi)=\mathcal{O}(1)$ according to TMAR, hence, $r$ terminates in constant time.
\end{proof}

Until now, we have considered DMAR only on a per-round basis, but it remains to show that the number of rounds before an instance of UMVRP-L is solved is also bounded.
The following analysis demonstrates this and will rely on Theorem~\ref{thm:Cluster-MAHR} and a known result about random walks.
The \textit{cover time} of a random walk on a connected, undirected graph $G(V,E)$ (denoted $\mathcal{C}(G)$) represents the worst-case expected time to visit every node of $G$ starting at any initial node.
It is known that $\mathcal{C}(G)\le 2|E|(|V|-1)$~\cite{motwani}.
We now show that the number of 
time steps before DMAR solves any solvable instance of UMVRP-L is bounded in expectation.

\begin{theorem}\label{thm:DMARexpected}
    For a solvable instance $G(V,E),O,S,\tau$ on $N$ nodes, DMAR completes all tasks in $\mathcal{O}(N^2)$ 
    time steps in expectation.
\end{theorem}

\begin{proof}
During a round of DMAR, if there are clusters in the TMAR phase, then some task will be completed in constant time by Theorem~\ref{thm:Cluster-MAHR} since the multiagent rollout algorithm is guaranteed to complete all tasks with our chosen heuristic~\cite{bertsekasMAR}.
Otherwise, each agent follows a random walk of length at least one. 
In the worst case, an agent may need to visit every node in the network to locate a task via random exploration. 
Consequently, by Theorem~\ref{thm:Cluster-MAHR}, the cover time bound, and since $|E|\le4N$, we conclude that the total number of 
time steps required to solve the instance is at most directly proportional to $8 N(N-1)=\mathcal{O}(N^2)$ in expectation.
\end{proof}

Although we have bounded the number of 
time steps required for DMAR to solve an instance of UMVRP-L in expectation, it is critical for any motion planner to show that it is \textit{probabilistically complete}, i.e. for any solvable instance $I$ of UMVRP-L, the probability that DMAR solves $I$ tends to unity as the number of rounds tends to infinity~\cite{probabilisticCompleteness}.
The following theorem demonstrates this.

\begin{theorem}\label{thm:probcomp}
    DMAR is a probabilisitically complete planner.
\end{theorem}
\begin{proof}
We argue from the squeeze principle.
Let random variable $T_{DMAR}$ be the length of an execution of DMAR in 
time steps given an arbitrary solvable instance of UMVRP-L.
As $\mathbb{E}[T_{DMAR}]$ is finite and positive by Theorem~\ref{thm:DMARexpected}, the sequence $(\mathbb{E}[T_{DMAR}]/t)_{t=1}^{\infty}$ converges to zero, as does the sequence $(0)_{t=1}^{\infty}$.
We use these facts to show the sequence $(\Pr{[T_{DMAR}\geq t]})_{t=1}^\infty$ also converges to zero.
Let $\epsilon>0$ be an arbitrary positive real number.
It suffices to show that there exits some $t'\in \mathbb{Z}^+$ such that if $t\geq t'$ then $\Pr{[T_{DMAR}\geq t]}<\epsilon$.
Choose $t'$ such that $\mathbb{E}[T_{DMAR}]/t<\epsilon$ for all $t\geq t'$.
Such a $t'$ must exist since $(\mathbb{E}[T_{DMAR}]/t)_{t=1}^{\infty}$ is decreasing and tends to zero.
By the Markov inequality, $
    0\le \Pr{[T_{DMAR}\geq t]}\le \mathbb{E}[T_{DMAR}]/t<\epsilon,\forall t\in \mathbb{Z}_{\geq t'}
$, and so $(\Pr{[T_{DMAR}\geq t}])_{t=1}^\infty$ also tends to zero.
Hence, under an execution of DMAR the probability that a reachable non-completed task persists tends to zero as the number of 
time steps tends to infinity, and so DMAR is probabilistically complete.
\end{proof}

\section{Experimental results}
\label{sec:experiments}
\label{sec:simulations-d}
We consider eight classes of instances of UMVRP-L; namely $10\times10,\ldots,80\times 80$ grid graphs.
For each class, we uniformly sample ten grids with 20\% of the nodes randomly designated as obstacles.
For each $\sqrt{N}\times \sqrt{N}$ grid, we consider a range of $k$-values (recall sensing radius), and three agent-to-task ratios (1:2,1:1, 2:1).
Agents and tasks are distributed uniformly at random over each instance for each ratio, and the number of agents is always $\sqrt{N}$.
Moreover, we disregard any topology that contains geographically isolated tasks, i.e. where the expected length of a random walk to visit any task is $\omega(N)$.
For all simulations, we fixed $\psi=8$, as preliminary simulations suggested this yields the best results.
We run ten independent runs per instance to account for randomness.
For each grid size and $k$-value combination, we report movement costs, wall-clock running times and number of clusters formed averaged over ten runs of ten instances for each of the three agent-to-task ratios.

For each of these combinations we run DMAR and a corresponding base policy (BP) called the \emph{greedy-exploration policy}.
Note that the algorithm that generates this policy differs from DMAR only in Line~\ref{alg:rollout} of Algorithm~\ref{alg:overview}, where instead of multiagent rollout, a greedy nearest-neighbor heuristic is used to compute a policy for each agent in $K$ 
 over $\mathcal{M}(K)$.
That is, the base policy still involves random exploration and mandates SOAC and LMA, but does not induce coordination.
The combinations described above result in approximately 50,000 individual simulations over 5,000 instances.
In this work we present results from $40\times40, 60\times 60$ and $80\times 80$ grids, but comprehensive results can be found in Appendix-C~\cite{weber2023distributed}.
Note that these experiments simulate the constraints of a distributed environment, and are not themselves decentralized.
Simulator source code can be found in Appendix-F~\cite{weber2023distributed}.

We observe from Figure~\ref{fig:results} (rows 1-3, left) that for each grid on $N\in \{40^2,60^2, 80^2\}$ nodes there exists a critical radius $k^*(N)$ such that for all $k\geq k^*(N)$, rollout outperforms the greedy-exploration base policy.
As $k$ decreases, we see that the performance of rollout degrades gracefully; and as $k$ increases, the total average running time seems to increase exponentially to account for increased online planning, illustrating the trade-off between scalability and solution quality.
From Figure~\ref{fig:results} (rows 1-3) we see that an effective range of radii exists that contains a special radius for which there are multiple clusters generated and the relative cost improvement is a factor of approximately two. 
At this special radius we achieve an average wall-clock running time that is a small fraction of the average time corresponding to the largest $k$ values we considered.
Note that the average running time for DMAR on an instance at any $k$ value is no larger than the running time of the centralized multiagent rollout algorithm, hence the observed improvement in running time over the largest $k$ values represents a lower bound when comparing DMAR to centralized MVRP.
For $k$ values considered beyond the effective range, we only observe a further constant factor relative cost improvement (approximately 3), even when the number of clusters generated approaches one (as in centralized MVRP), and hence only a small benefit is gained at the expense of significantly higher running times. 
In Figure~\ref{fig:results} (bottom), we plot an approximation of $k^*(N)$ as obtained from our samples, and juxtapose several slow-growing functions.
We see that $k^*(N)$ is closely approximated by $\log_2^*(N)$.
Consequently, we mark our effective range bounds in Figure~\ref{fig:results} (rows 1-3) (indicated by the shaded box) as $2\log^*_2(N)$ to $3\log^*_2(N)$ (i.e. 8-12).

\begin{figure}
     \centering
    \begin{subfigure}[b]{0.95\textwidth/4}
         \centering
         \includegraphics[width=\textwidth]{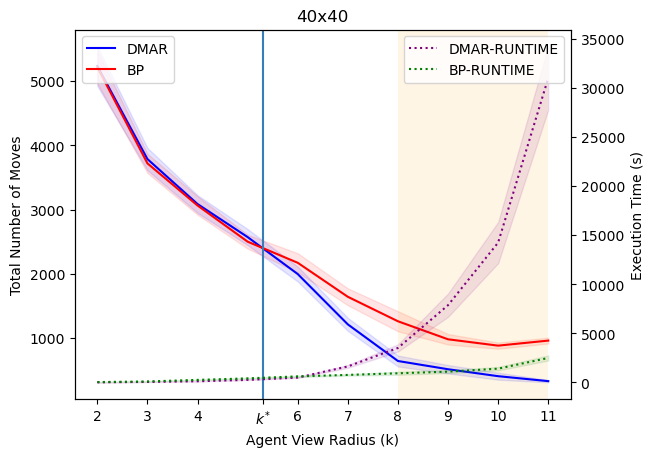}
         
     \end{subfigure}
     \hspace{-2.3mm}
     \begin{subfigure}[b]{0.9\textwidth/4}
         \centering
         \includegraphics[scale=0.255]{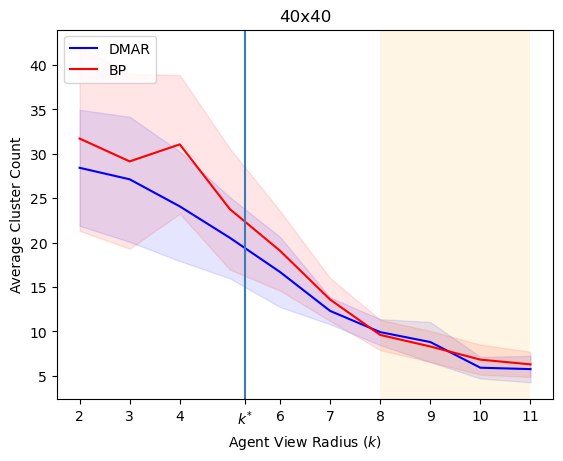}
         
     \end{subfigure}
     \hspace{-2.3mm}
     \begin{subfigure}[b]{0.95\textwidth/4}
         \centering
         \includegraphics[width=\textwidth]{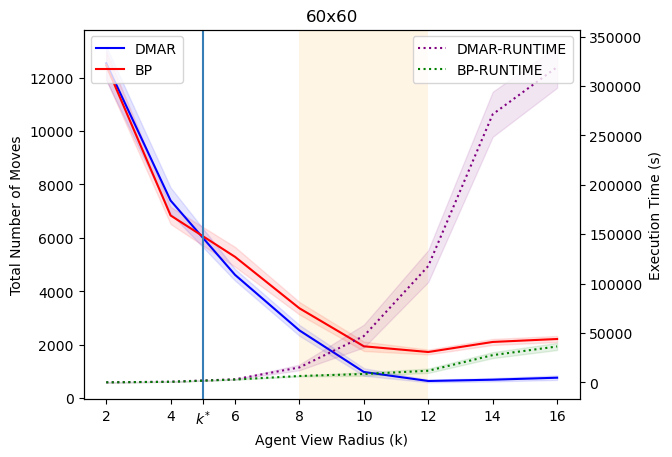}
       
     \end{subfigure}
     \hspace{-2.3mm}
     \begin{subfigure}[b]{0.9\textwidth/4}
         \centering
         \includegraphics[scale=0.255]{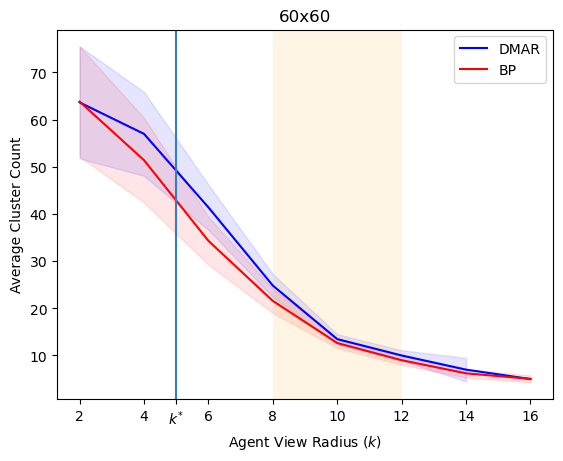}
         
     \end{subfigure}
     \hspace{-2.3mm}
    \begin{subfigure}[b]{0.95\textwidth/4}
         \centering
         \includegraphics[width=\textwidth]{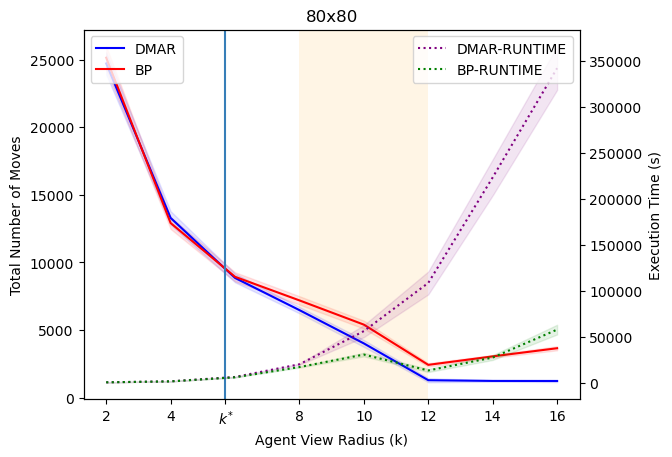}
      
     \end{subfigure}
    \hspace{-2.3mm}
    \begin{subfigure}[b]{0.9\textwidth/4}
         \centering
         \includegraphics[scale=0.255]{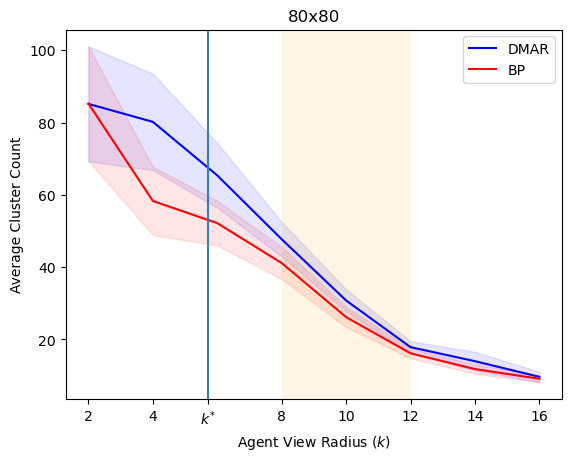}
         
     \end{subfigure}
     \begin{subfigure}[b]{0.9\textwidth/4}
        \includegraphics[scale=0.255]{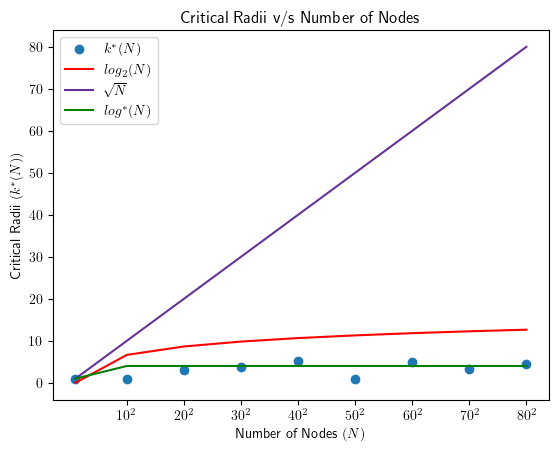}
       
        \label{fig:critical_ratios}
     \end{subfigure}
        \caption{
        (Rows 1-3, left) Left vertical axes show the average cost of greedy-exploration base policy vs average cost of DMAR. Right vertical axes show average running time in seconds for DMAR and base policy. Critical radii are marked as $k^*$; shaded orange boxes show effective ranges. 
        (Rows 1-3, right) Show average number of clusters from base policy vs those from DMAR. 
        95\% confidence intervals are shown by shaded regions around respective means. (Bottom) Sampled critical radii function $k^*(N)$.
       }
        \label{fig:results}
\end{figure}

\paragraph{\textbf{Physics-based Simulations and Robotics Experiments}}
To address the question of whether DMAR is applicable to real-world scenarios, we have created a proof-of-concept implementation that adapts DMAR to continuous space using the Robotarium platform~\cite{robotarium}.
These simulations capture the dynamics of differential drive-based robots; that is, each robot's movement is generated by actuation of motors, the dynamics of which are described by a set of ordinary differential equations.
In this adaptation, planning by each robot occurs in discrete space (9$\times$9 grids), however, once a discrete control sequence is computed, each control is then mapped to a continuous space velocity vector, after which the new robot position is realized by the actuators.
We took measures to avoid collisions and simulated the distributed constraints of our model on a centralized server (see Appendix-D~\cite{weber2023distributed}).

Due to the size constraints on the arena, we limit our simulations to 30 randomly generated environments of size \qty{2}{\meter}$\times$\qty{2}{\meter}, each with a varying number of tasks and 7 agents. 
For each instance, we ran 10 executions and obtained average costs.
In addition, we also varied the sensing radius $k$ as $20,40,60,80$ cm and $\infty$.
For all simulations, we fixed $\psi=4$.
The results of these experiments are presented in Figure~\ref{fig:results-phys} (left).
Here DMAR outperforms the greedy-exploration policy regardless of $k$, although its performance degrades as $k$ decreases.
We see that as the radius increases, the number of clusters decreases.
We also ran experiments
on physical robots using the Robotarium platform on many instances (Figure~\ref{fig:results-phys} (right)).

\begin{figure}
     \centering
     \begin{subfigure}[b]{0.95\textwidth/4}
         \centering
         \includegraphics[width=\textwidth]{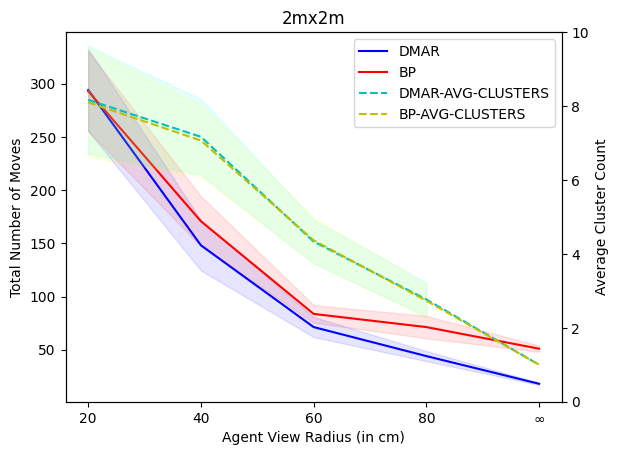}
   
         \label{fig:robotarium}
     \end{subfigure}
     \hspace{-1mm}
     \begin{subfigure}[b]{0.9\textwidth/4}
         \centering
         \includegraphics[scale=0.08]{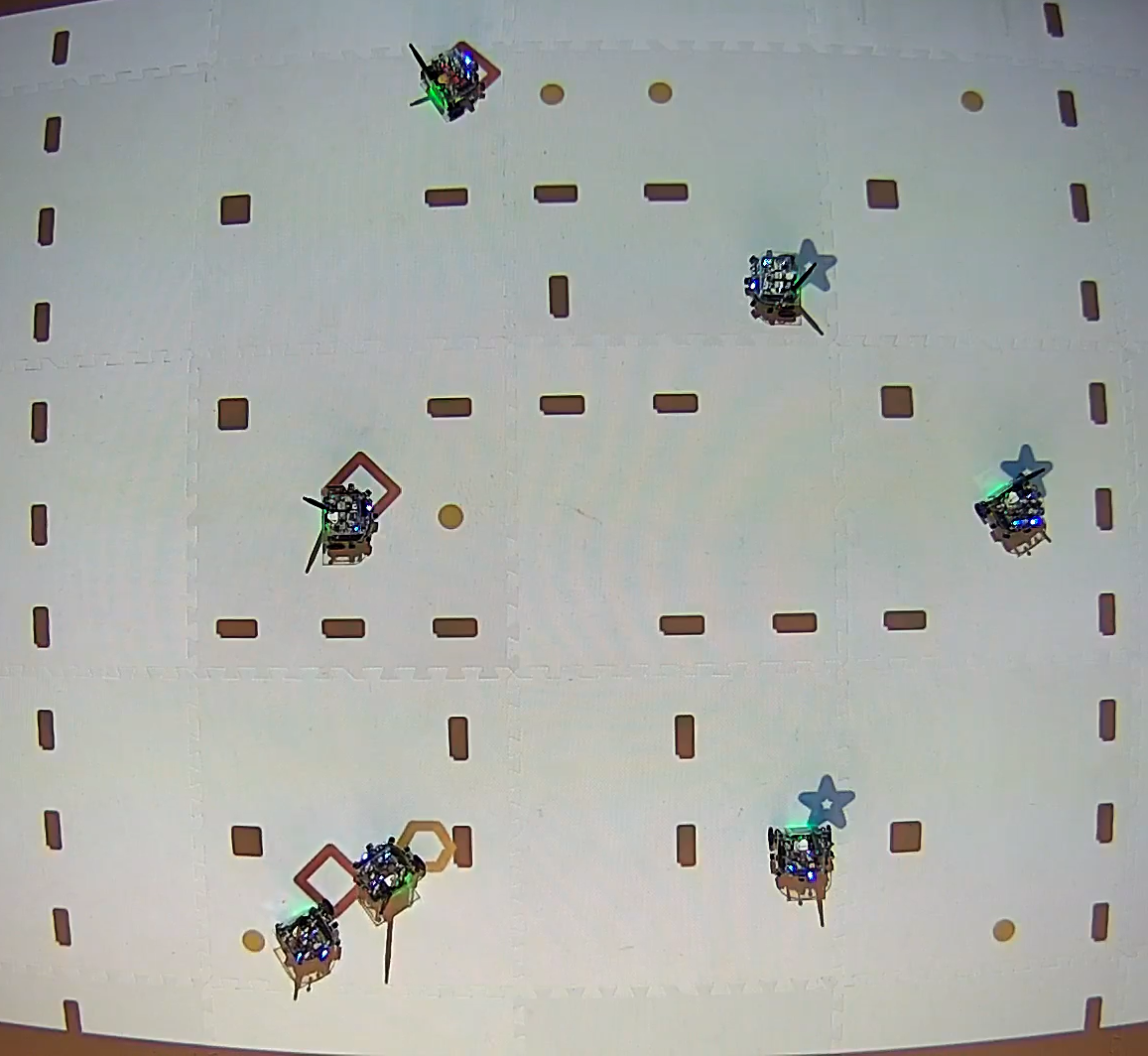}
         
         \label{fig:physical_experiments}
     \end{subfigure}
        \caption{
        (Left) Physics-based simulation on \qty{2}{\meter}$\times$\qty{2}{\meter} environments. Left vertical axis shows average solution cost, right axis shows average number of clusters. 95\% confidence intervals are given by shaded regions.
        (Right) DMAR execution in Robotarium. Small dots are tasks, boxes are obstacles. Star, diamond, hexagon indicate cluster membership. }
        \label{fig:results-phys}
\end{figure}

\section{Discussion}
\label{sec:conclusions}
We have presented a distributed computing approach for applying rollout to a generalization of the multivehicle routing problem where agents operate with limited local sensing capabilities and possess no a priori knowledge of the network.
This study allowed us to examine the role of communication in rollout, and the effect that limiting it has on performance.
Our approach produces quality solutions already for small
local sensing radii.
We have shown that there exists an (effectively constant) critical sensing radius beyond which rollout outperforms a greedy strategy for all larger radii.
Moreover, there exists a range of effective radii where the relative cost improvement over the base policy is approximately a factor of two, but communication is still significantly limited.
In this effective range we also observed radii where the running time of DMAR was exponentially smaller with only a constant factor performance detriment as compared to larger radii.
This implies the superior scalability of our approach, and illustrates its trade-off with performance.
We conclude that the user may choose a radius between 8 and 12 and expect substantial relative cost improvement and fast execution without knowledge of the network size, and that our approach is conceptually adaptable to continuous space and robust to sensor noise, augmenting its real-world applicability.

Regarding the limitations of this work, real-world distributed systems rarely behave synchronously, which we do not address here. Moreover, arbitrary physical environments may be noisy or change over time.
As such, future work includes adapting DMAR to stochastic and dynamic environments, and to weaker distributed models. 
Nevertheless, our approach may apply more broadly to other such decentralized multiagent problems, as vehicle routing is a fundamental primitive of many multiagent algorithms, e.g.~\cite{uav,exploration,DREXL2013275,platoon1,platoon2}.
Lastly, our target applications (recall minefield disarmament and post-disaster search and rescue) suggest that DMAR may produce a considerable positive societal impact.




\begin{acks}
    Special thanks to Ted Pavlic for guidance, and to Anya Chaturvedi and Joseph Briones for technical and organizational support. 
    Research partially supported under
NSF CCF-1637393 and CCF-1733680 and DoD-ARO MURI
No.W911NF-19-1-0233 awards.
\end{acks}



\bibliographystyle{ACM-Reference-Format} 
\bibliography{refs}


\newpage
\appendix

\section{Alternative Views}
For our experiments, we chose the view for an agent $a$ to be the subgraph of $G$ induced by all nodes reachable by at most $k$ hops from the position of $a$, $\rho(a)$, on the grid---we call this view (i).
However, other definitions of a view may make more sense in different contexts.
In this section we discuss two alternative views:
(ii) ($k$-hop-reduced) The view $\nu(a)$ of an agent $a$ is given by the subgraph of $G$ induced by all nodes reachable in at most $k$ hops from $\rho(a)$ on a graph $G'$ that corresponds to $G$ with all obstacle nodes removed;
(iii) (line of sight) The view $\nu(a)$ of an agent $a$ is given by the subgraph of $G$ induced by all nodes $v$ reachable in at most $k$ hops from $\rho(a)$ on a graph $G'$ that corresponds to $G$ with all obstacle nodes removed, and such that a line segment connecting the coordinates of $\rho(a)$ to the coordinates of $v$ can be drawn that does not intersect an obstacle.
Note that (i) is the least restrictive of the three proposed views. Thus, as our theoretical results from Section~\ref{sec:algorithms} reflect worst-case running time as a function of the largest cluster size, they apply to views (ii) and (iii) as well since these definitions can only make the largest cluster size smaller.
We also remark that the choice of view is application dependent.
For example, with definitions (i) and (ii), agents may see through obstacles, which may make sense depending on the application, whereas (iii) may be more suitable for a maze environment.

Experimentally, we reran our experiments from Section~\ref{sec:experiments} on the set of $40\times40$ and $60\times 60$ grids using definition (ii).
The results are summarized in Figure~\ref{ldfs}.
The more restrictive view definition yielded smaller clusters, and hence higher solution costs as compared with definition (i), although the shapes of the curves remain consistent with our results using definition (i). 
Moreover, the running time was overall lower.
We conclude that using more restrictive view definitions effectively shifts the cost curves up, while shifting the computational cost curves down.
Both are attributed to the fact that less information is available under the restricted view, the former resulting from a dampened cost benefit, and the latter due to there being less information on which to perform computations.
We did not run experiments using view definition (iii), but we predict that these curve shifts would be even more pronounced.

\begin{figure}
    \centering
    \begin{subfigure}[b]{0.98\textwidth/4}
    \includegraphics[width=\textwidth]{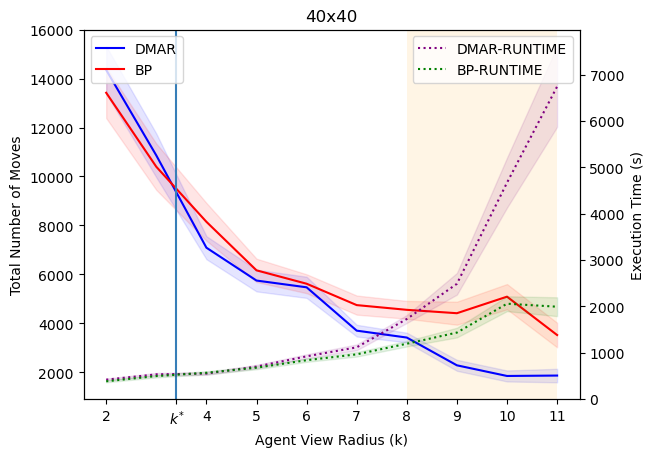}
    \caption{}
    \label{fig:ldfs-bp-dmar}
    \end{subfigure}
    \hfill
    \begin{subfigure}[b]{0.85\textwidth/4}
    \includegraphics[width=\textwidth]{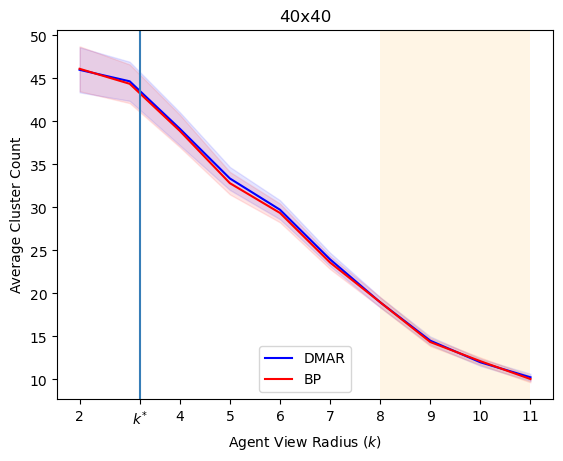}
    \caption{}
    \label{fig:ldfs-acc}
    \end{subfigure}
    \hfill
    \begin{subfigure}[b]{0.98\textwidth/4}
    \includegraphics[width=\textwidth]{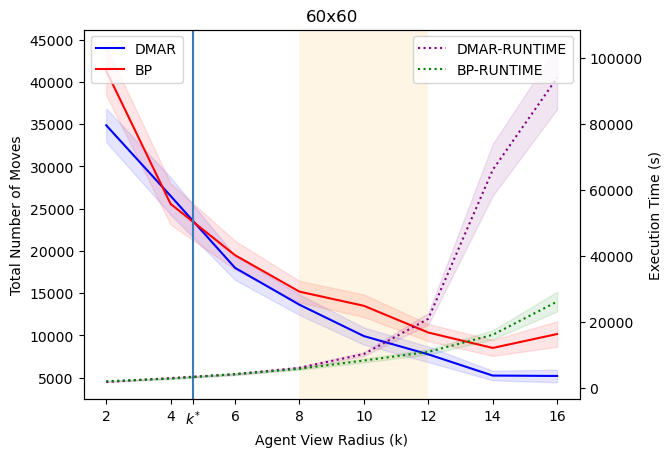}
    \caption{}
    \label{fig:ldfs-bp-dmar-60}
    \end{subfigure}
    \hfill
    \begin{subfigure}[b]{0.85\textwidth/4}
    \includegraphics[width=\textwidth]{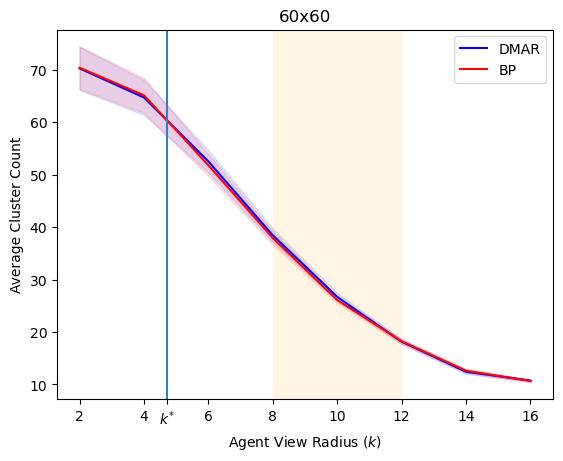}
    \caption{}
    \label{fig:ldfs-acc-60}
    \end{subfigure}
    \caption{Results for a $k$-hop-reduced view definition on $40\times40, 60\times60$ grids. (a,c) Left vertical axis shows average cost of greedy-exploration base policy (red) vs average cost of DMAR (blue). Right vertical axis shows average running time in seconds for DMAR and base policy. Critical radius is marked as $k^*$; shaded orange boxes show effective ranges. 
        (b,d) Show average number of clusters from base policy (red) vs those from DMAR (blue). 
        95\% confidence intervals shown by shaded regions around respective means.}
    \label{ldfs}
\end{figure}

\section{DMAR variant with enforced cost-improvement property}

In this section we present a modification of DMAR (and its corresponding base policy) that contains additional mechanisms that guarantee cost improvement---although at a cost. 
This is motivated by the fact that global solution quality is dependent on the sequence of clusters formed, and therefore cost improvement obtained from solving clusters using DMAR can be difficult to measure.

In our modified version of DMAR with a guaranteed relative cost-improvement property (DMAR-GCI), after all tasks in a cluster $K$ are completed, the agents of $K$ gather at a common position within the cluster view (we call this a \emph{depot}) regardless of whether rollout or the greedy-exploration base policy (denoted as BP-GCI) is executed.
More specifically, let $\mathcal{R}$ be the set of control sequences for $K$ computed by the leader agent $\ell$ during an execution of TMAR.
Also, denote $\rho_f(\mathcal{R}^a)$ as the final position of an agent $a$ as obtained from executing the control sequence $\mathcal{R}^a$.
Agent $\ell$ then chooses a cluster depot at its own final position (i.e. $\rho_f(\mathcal{R}^{\ell}))$ and appends additional movements to each control sequence such that for each agent $a\in K$, the control sequence corresponding to the shortest path from $\rho_f(\mathcal{R}^a)$ to $\rho_f(\mathcal{R}^{\ell})$ is appended to the end of $\mathcal{R}^a$ (followed by the required number of \textsc{wait} controls such that the control sequence length is $\lambda(\psi)$.)
Agent $\ell$ then initiates the broadcast of $\mathcal{R}$ as normal.
As such, the final position of all agents in $K$ is necessarily $\rho_f(\mathcal{R}^{\ell})$.
In this scheme, we also prevent exploring agents from interfering with existing clusters, which will be a useful for our analysis.
If an agent $a\in K$ is visiting a task location $v\in \nu(K)$, agent $a$ leaves a special ``token'' at $v$.
If then $v$ appears in the view of an agent $a'\notin K$, $a'$ will still recognize $v$ as a task location and treat it as such (even though visiting $v$ does not increase the number of completed tasks).
We refer to this as the \emph{token mechanism}.

The following theorem specifies a cost-improvement guarantee for DMAR-GCI, and relies on the cost-improvement properties of multiagent rollout as seen in~\cite{bertsekasMAR}.
\begin{theorem}
    \label{thm:gci}
    For any instance $I$ of UMVRP-L, the total cost generated by DMAR-GCI from an execution on $I$ is no larger than the total cost generated by an execution of BP-GCI on $I$.
\end{theorem}
\begin{proof}
We claim by induction on the number of rounds of DMAR-GCI/BP-GCI that under an identical source of randomness, parallel executions off DMAR-GCI and BP-GCI yield identical clustering sequences.
We demonstrate the following invariant: 
Under an identical source of randomness, at the end of each round of DMAR-GCI/BP-GCI, each individual agent occupies the same respective node in both executions. 
As the source of randomness is common, the initial clusters formed will be identical in both executions. 
Moreover, at the end of the first round, the positions of all agents that are members of clusters will be the same in both executions, as all such agents will be gathered at their respective depots.
By the token mechanism, each agent that is not a member of a cluster will also end the first round of DMAR-GCI/BP-GCI in the same location in both instances, as differences in the order in which BP-GCI and DMAR-GCI complete tasks will not influence the trajectories of exploring agents.
Hence, in both executions, each agent occupies the same respective position at the end of the first round.
For an arbitrary round, if we assume agents occupy the same respective positions in both runs of DMAR-GCI/BP-GCI, then an analogous argument demonstrates that the invariant holds for the next round, as the initial agent positions are arbitrary and identical in both executions.

The policy generated by BP-GCI for each agent is given by its movement sequences, which are a combination of random exploration movements, as well as movements within clusters dictated by the greedy nearest-neighbor base heuristic.
As these random movements are identical regardless of whether generated by DMAR-GCI and BP-GCI, and the greedy nearest-neighbor base heuristic applied within a cluster is \emph{sequentially consistent} (i.e. the control selected is dependent only on the current state; see~\cite{bertsekasMAR}), BP-GCI is also a sequentially consistent base policy for DMAR-GCI.
From~\cite{bertsekasMAR}, this property suffices to conclude that DMAR-GCI can perform no worse than BP-GCI.
\end{proof}

As the exploration costs are identical between DMAR-GCI and BP-GCI, they can be discounted, and hence we may directly compare the performance of our rollout scheme to the base policy without accounting for exploration.
Note that DMAR-GCI is purely an analytical tool that is not meant to be used in practice.
It's purpose is to allow us to compare DMAR with the base policy directly, without accounting for volatile random exploration.
In Appendix~\ref{appendix:additional_simulations}, we show that DMAR-GCI approximates DMAR (overhead is added for agents congregating at the depots), and hence the DMAR-GCI versus BP-GCI comparison may provide a better sense of the cost-improvement effect of DMAR for UMVRP-L.

\section{Comprehensive simulation results}
\label{appendix:additional_simulations}
The full results (i.e. for each grid size $10\times10,20\times20,\ldots,80\times80$) from comparing the average total costs generated by both DMAR and BP, as well as the corresponding average runtimes are presented in Figure~\ref{fig:results-appendix}. 
These results fall in line with the analysis given in Section~\ref{sec:simulations-d}; that is, the critical radius is approximately 4, and the effective region is between $k=8$ and $k=12$, wherein there exists a radius for which we observe an approximate 2-factor relative cost improvement.
Moreover, beyond the effective range, the running time of DMAR increases seemingly exponentially.
The only exceptions are the $10\times10$ grids, where DMAR outperforms the base policy for all radii tested on average.
However, these small instances do not represent the trends we find from scaling the instances.

\begin{figure}
     \centering
     \begin{subfigure}[b]{0.95\textwidth/4}
         \centering
         \includegraphics[width=\textwidth]{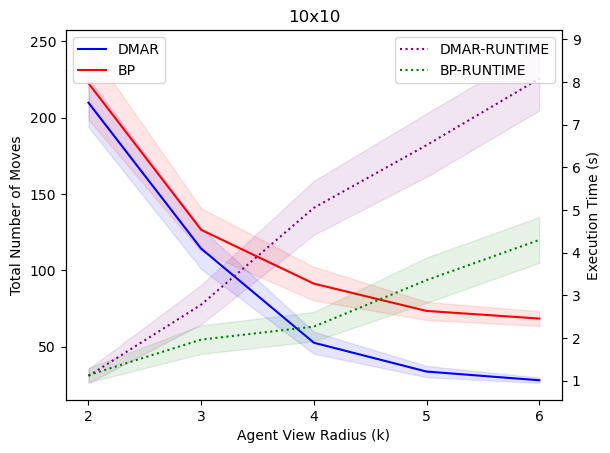}
         \caption{}
     \end{subfigure}
     \hspace{-2.3mm}
     \begin{subfigure}[b]{0.95\textwidth/4}
         \centering
         \includegraphics[width=\textwidth]{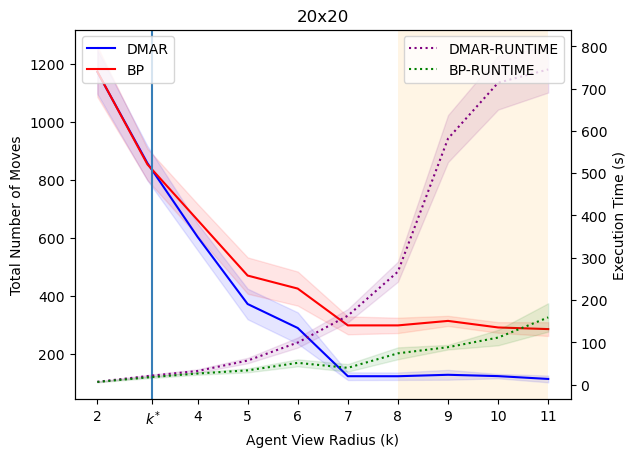}
         \caption{}
     \end{subfigure}
     \hspace{-2.3mm}
     \begin{subfigure}[b]{0.95\textwidth/4}
         \centering
         \includegraphics[width=\textwidth]{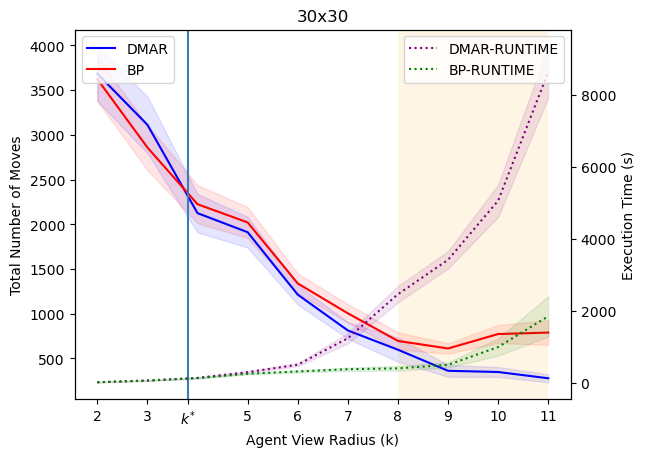}
         \caption{}
     \end{subfigure}
    \hspace{-2.3mm}
    \begin{subfigure}[b]{0.95\textwidth/4}
         \centering
         \includegraphics[width=\textwidth]{40-BP_vs_DMAR.png}
         \caption{}
     \end{subfigure}
    \hspace{-2.3mm}
    \begin{subfigure}[b]{0.95\textwidth/4}
         \centering
         \includegraphics[width=\textwidth]{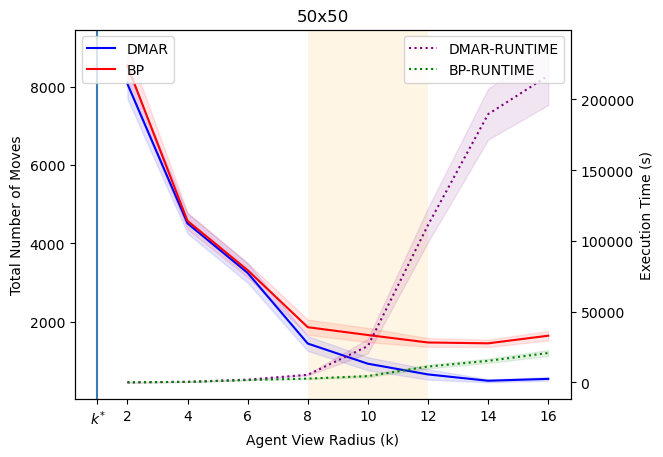}
         \caption{}
     \end{subfigure}
     \hspace{-2.3mm}
     \begin{subfigure}[b]{0.95\textwidth/4}
         \centering
         \includegraphics[width=\textwidth]{60-BP_vs_DMAR.png}
         \caption{}
     \end{subfigure}
     \hspace{-2.3mm}
     \begin{subfigure}[b]{0.95\textwidth/4}
         \centering
         \includegraphics[width=\textwidth]{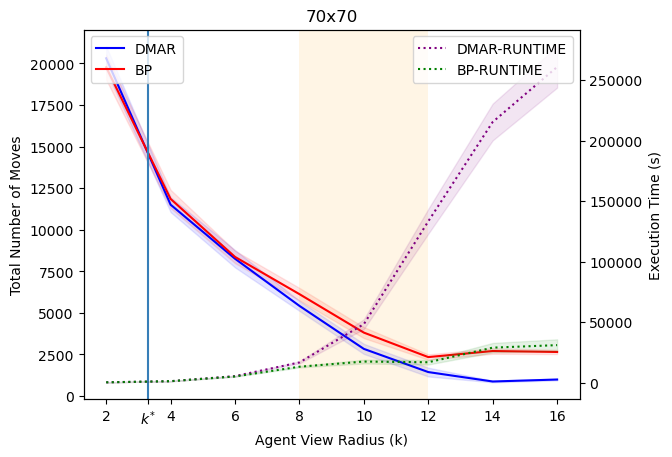}
         \caption{}
     \end{subfigure}
    \hspace{-2.3mm}
    \begin{subfigure}[b]{0.95\textwidth/4}
         \centering
         \includegraphics[width=\textwidth]{80-BP_vs_DMAR.png}
         \caption{}
     \end{subfigure}
        \caption{Experimental results. 
        (a-h, left vertical axis) Discrete-space simulation plots show the average cost of the greedy-exploration base policy (red curves) vs average cost of DMAR (blue curves) on $10\times10$, $20\times20$, $\ldots$, $80\times80$ grids. (a-h, right vertical axis). Curves correspond to average running time in seconds for both DMAR and the base policy. Critical radii are marked as $k^*$ are marked, effective ranges are represented by the shaded orange boxes. 95\% confidence intervals are given by shaded regions around their respective means.
       }
        \label{fig:results-appendix}
\end{figure}

\subsection{Simulation results for DMAR-GCI}
In Figure~\ref{fig:results-depot}, we present the results for comparing the average total costs of DMAR-GCI and BP-GCI. 
Each data point ($k$-value for a specific grid size) is an average of 50 or more individual simulations over various topologies and agent-to-task rations.
We validate Theorem~\ref{thm:gci} by observing that the average total costs obtained from DMAR-GCI are no larger than those of BP-GCI for any $k$-value or grid size considered.
We see that in comparison to DMAR (Figure~\ref{fig:results-appendix}), the average costs generated by DMAR-GCI are slightly larger than those generated by DMAR, although this difference becomes less pronounced as the grid size increases.
We also observe that the proportion of total average cost attributed to random exploration moves is identical for both algorithms.
Consequently, in Figure~\ref{fig:noexp}, we present the results from Figure~\ref{fig:results-depot} with the exploration costs subtracted out to better visualize the cost improvement obtained by DMAR-GCI. 
As a result, the costs appear to increase with $k$, although this is only due to the fact that the proportion of average total cost attributed to random exploration decreases with $k$ (since more of the environment is visible to agents).
Nevertheless, in the effective range ($k=8$ to $k=12$) and for each grid size, there exists a radius for which we observe a cost improvement factor between 1.7 and 2.
As DMAR-GCI approximates the costs generated by DMAR, we cite DMAR-GCI to further support the cost-improvement claims with respect to DMAR.

\begin{figure}
     \centering
     \begin{subfigure}[b]{0.95\textwidth/4}
         \centering
         \includegraphics[width=\textwidth]{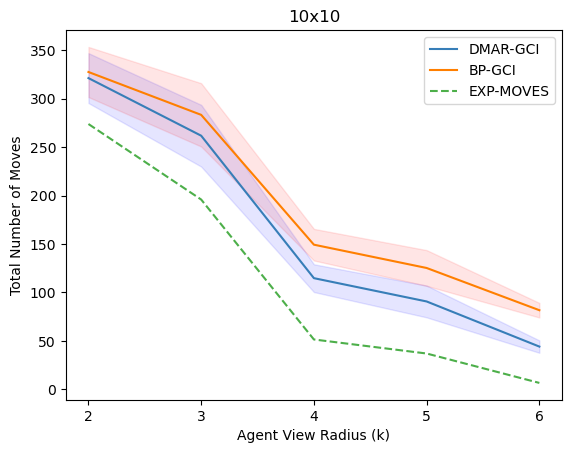}
         \caption{}
     \end{subfigure}
     \hspace{-2.3mm}
     \begin{subfigure}[b]{0.95\textwidth/4}
         \centering
         \includegraphics[width=\textwidth]{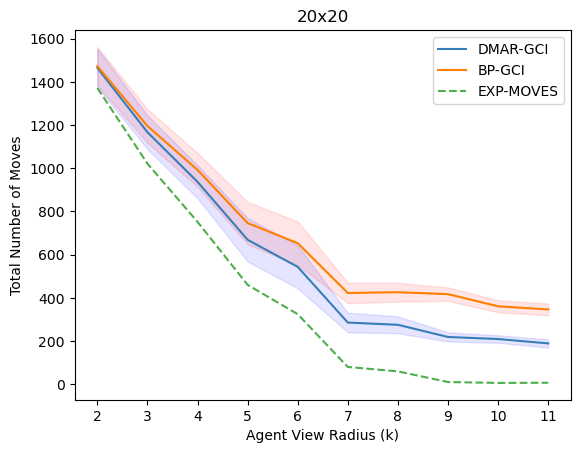}
         \caption{}
     \end{subfigure}
     \hspace{-2.3mm}
     \begin{subfigure}[b]{0.95\textwidth/4}
         \centering
         \includegraphics[width=\textwidth]{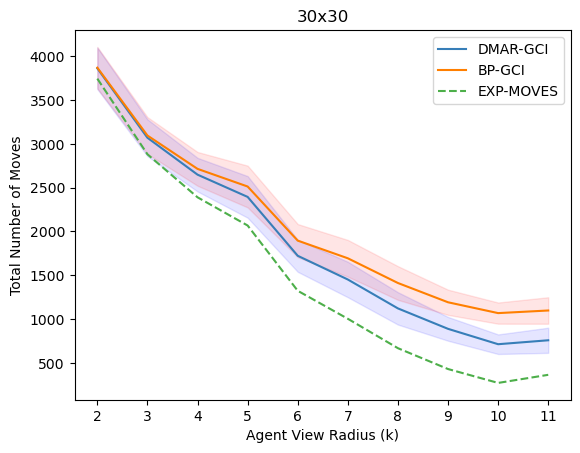}
         \caption{}
     \end{subfigure}
    \hspace{-2.3mm}
    \begin{subfigure}[b]{0.95\textwidth/4}
         \centering
         \includegraphics[width=\textwidth]{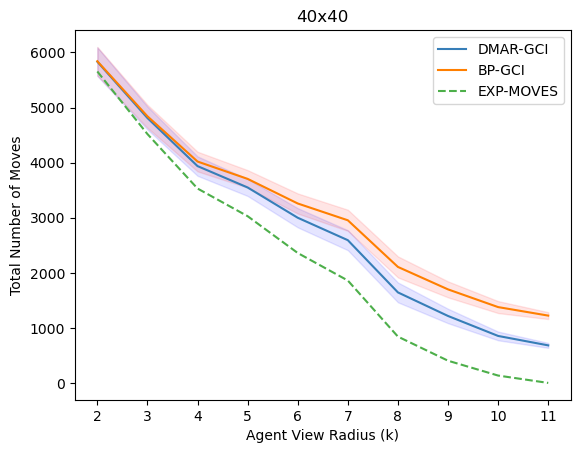}
         \caption{}
     \end{subfigure}
    \hspace{-2.3mm}
    \begin{subfigure}[b]{0.95\textwidth/4}
         \centering
         \includegraphics[width=\textwidth]{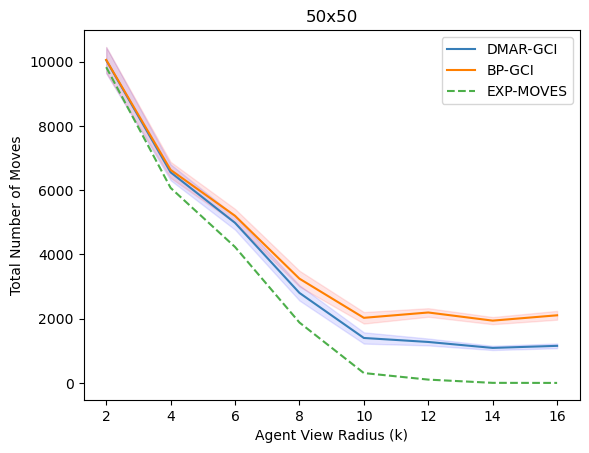}
         \caption{}
     \end{subfigure}
     \hspace{-2.3mm}
     \begin{subfigure}[b]{0.95\textwidth/4}
         \centering
         \includegraphics[width=\textwidth]{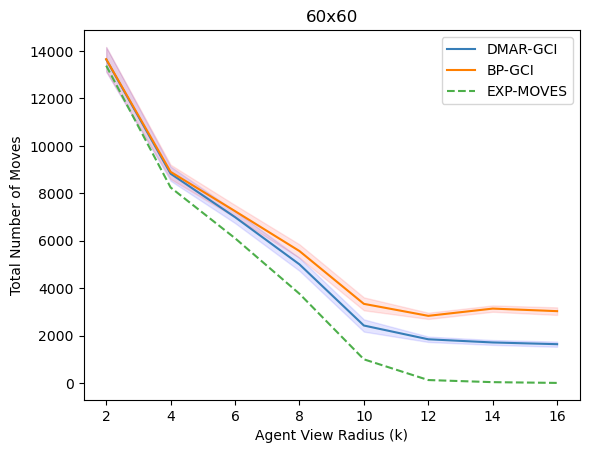}
         \caption{}
     \end{subfigure}
     \hspace{-2.3mm}
     \begin{subfigure}[b]{0.95\textwidth/4}
         \centering
         \includegraphics[width=\textwidth]{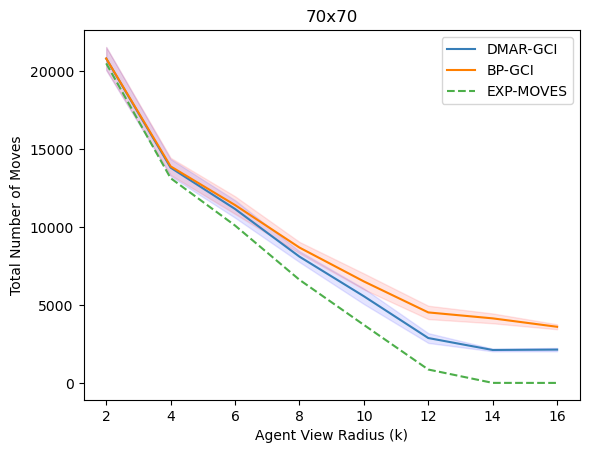}
         \caption{}
     \end{subfigure}
    \hspace{-2.3mm}
    \begin{subfigure}[b]{0.95\textwidth/4}
         \centering
         \includegraphics[width=\textwidth]{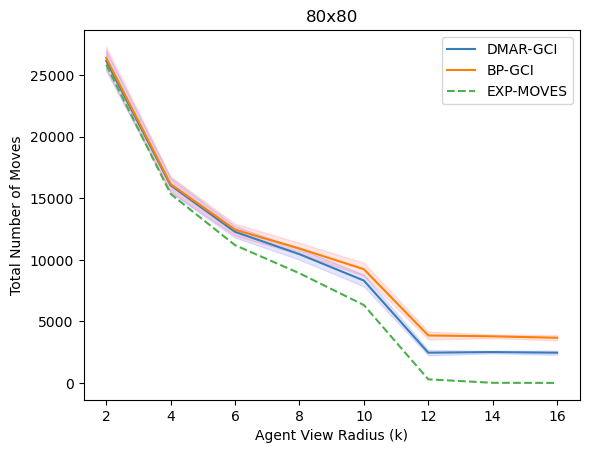}
         \caption{}
     \end{subfigure}
        \caption{Experimental results. 
        Discrete-space simulation plots show the average cost of BP-GCI (red curves) vs average cost of DMAR-GCI (blue curves) on $10\times10$, $20\times20$, $\ldots$, $80\times80$ grids. The green dashed curve represents the proportion of total cost that is attributed to exploration moves. This is identical for both algorithms. 95\% confidence intervals are given by shaded regions around their respective means.
       }
        \label{fig:results-depot}
\end{figure}

\begin{figure}
     \centering
     \begin{subfigure}[b]{0.95\textwidth/4}
         \centering
         \includegraphics[width=\textwidth]{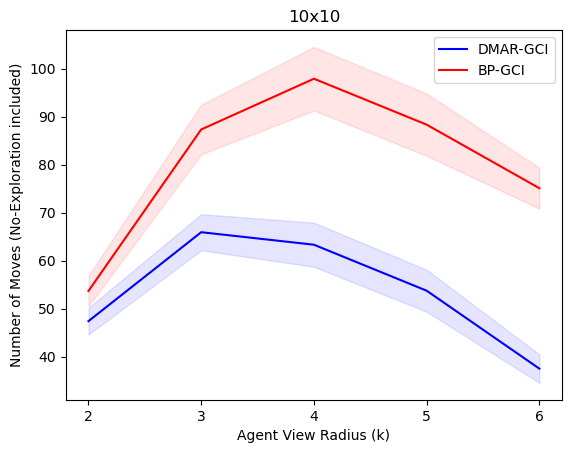}
         \caption{}
     \end{subfigure}
     \hspace{-2.3mm}
     \begin{subfigure}[b]{0.95\textwidth/4}
         \centering
         \includegraphics[width=\textwidth]{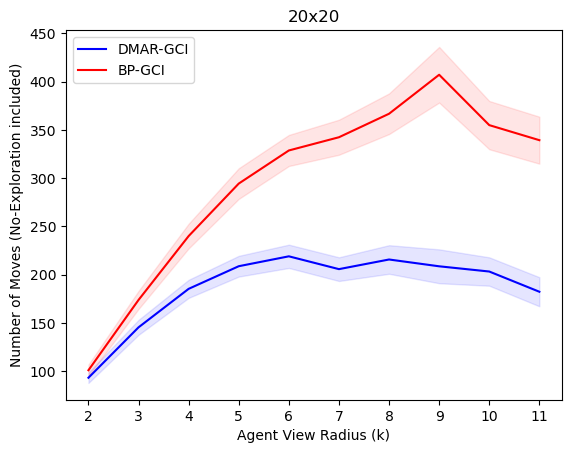}
         \caption{}
     \end{subfigure}
     \hspace{-2.3mm}
     \begin{subfigure}[b]{0.95\textwidth/4}
         \centering
         \includegraphics[width=\textwidth]{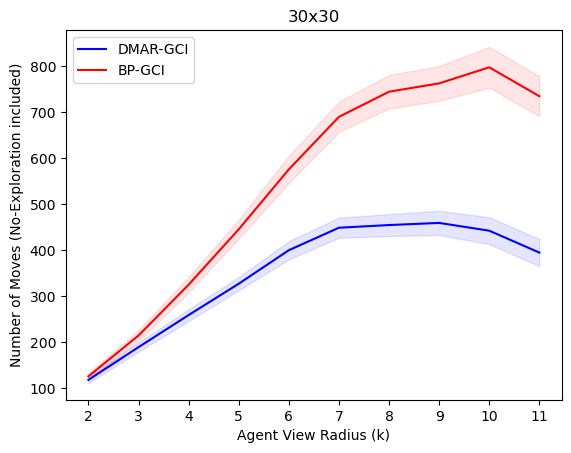}
         \caption{}
     \end{subfigure}
    \hspace{-2.3mm}
    \begin{subfigure}[b]{0.95\textwidth/4}
         \centering
         \includegraphics[width=\textwidth]{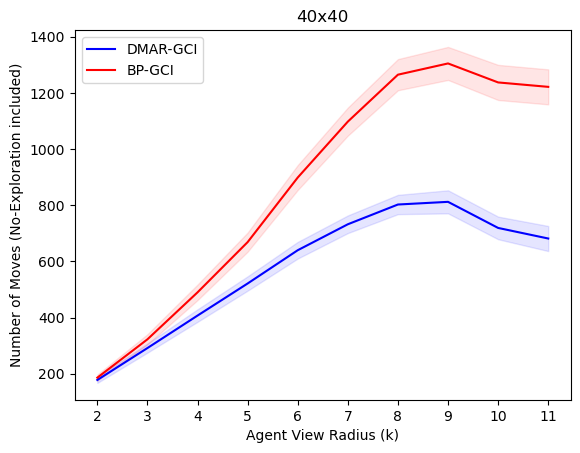}
         \caption{}
     \end{subfigure}
    \hspace{-2.3mm}
    \begin{subfigure}[b]{0.95\textwidth/4}
         \centering
         \includegraphics[width=\textwidth]{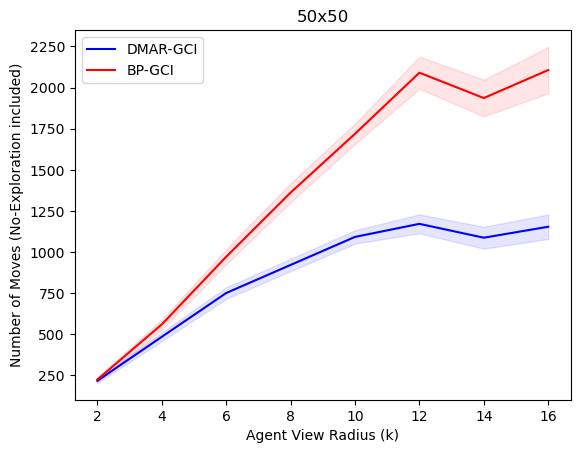}
         \caption{}
     \end{subfigure}
     \hspace{-2.3mm}
     \begin{subfigure}[b]{0.95\textwidth/4}
         \centering
         \includegraphics[width=\textwidth]{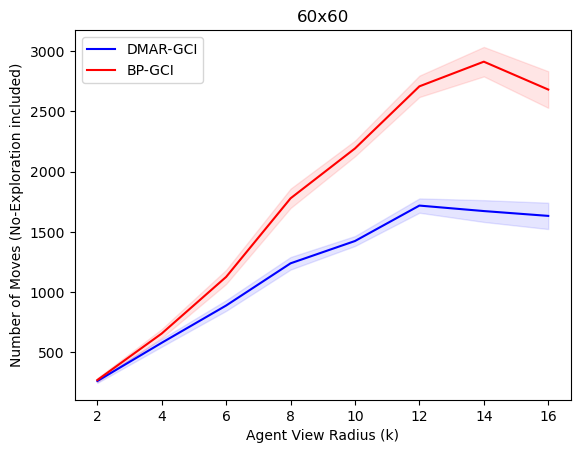}
         \caption{}
     \end{subfigure}
     \hspace{-2.3mm}
     \begin{subfigure}[b]{0.95\textwidth/4}
         \centering
         \includegraphics[width=\textwidth]{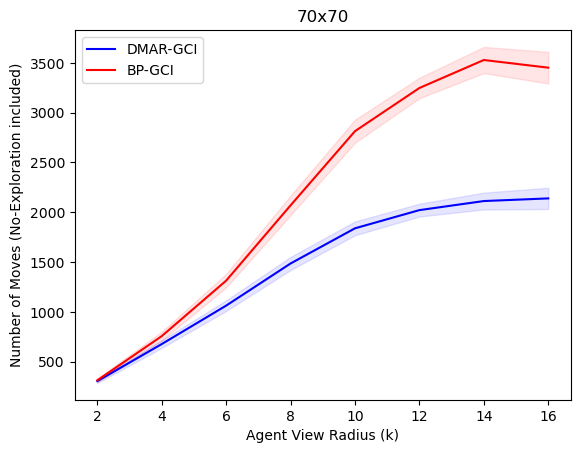}
         \caption{}
     \end{subfigure}
    \hspace{-2.3mm}
    \begin{subfigure}[b]{0.95\textwidth/4}
         \centering
         \includegraphics[width=\textwidth]{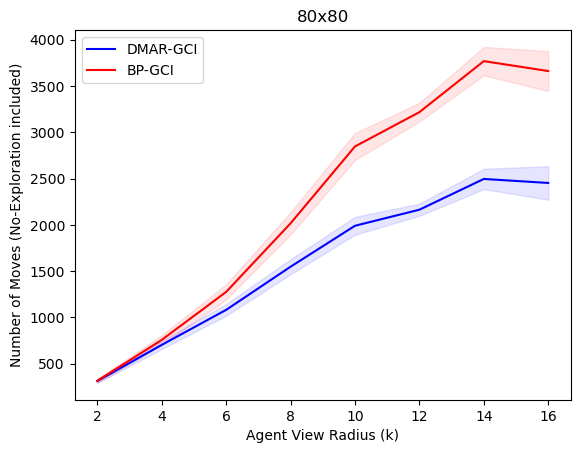}
         \caption{}
     \end{subfigure}
        \caption{Experimental results. 
        Simulation plots show the average cost of BP-GCI (red curves) vs average cost of DMAR-GCI (blue curves) on $10\times10$, $20\times20$, $\ldots$, $80\times80$ grids. The common exploration cost incurred by both algorithms has been subtracted out. 95\% confidence intervals are given by shaded regions around their respective means.
       }
        \label{fig:noexp}
\end{figure}

\section{Experimental methods: robotics specifications}
\label{appendix:robotarium}
Each robot in the Robotarium is approximately \qty{11}{\centi\meter}$\times$\qty{10}{\centi\meter} and operates in an arena (Robotarium testbed) of \qty{3}{\meter}$\times$\qty{2}{\meter}. The robots are initialized in the arena in accordance to the Cartesian coordinate plane with center of the arena being the origin of the plane. To transform the discrete-space version of DMAR to a continuous setting, we have approximated the arena (of size \qty{2}{\meter}$\times$\qty{2}{\meter}) as a $9\times9$ grid of cells with each cell of size \qty{20}{\centi\meter}$\times$\qty{20}{\centi\meter}. Robots do not strictly move as per the grid, but only use the center points of each cell as a guiding point when planning and moving. The actual motion of the robots is controlled using a hybrid controller, which first moves the robot on a linear path towards the goal position and then changes its orientation by rotating once at the goal position to point it in the desired direction. This provides for a smooth motion and allows for movement where robots can move past each other easily. The robots are based on a differential drive model and use single integrator dynamics to define a controller for the robot. The exact derivation and treatment of the robot dynamics are given in~\cite{robotarium}. These dynamics incorporate robust barrier functions that ensure that inter-robot and wall-robot collisions do not occur.

For the simulations to function on actual physical robots, two important requirements arise that the DMAR algorithm must address:
\begin{enumerate}
    \item \textbf{Collision barriers}: agents may not physically bump into each other while moving around in the arena.
    \item \textbf{Collision avoidance}: no agent occupies the same space as another.
\end{enumerate}

The barrier functions of the Robotarium platform keep the physical robots safe from damage, hence this requirement is met. Now, unlike the discrete simulations where robots may occupy the same location at a given time on the grid space, in a physical environment, two robots cannot occupy the same space at a given time. This required the use of a collision avoidance mechanism, which we have designed to be least obtrusive to the original algorithm's movement decisions.

All the phases of the Algorithm \ref{alg:overview} progress as before, except we augment the last phase of Execute Movement. The mechanism deployed to achieve collision avoidance is described in Algorithm \ref{alg:collision:avoid}. Note that although the implementation is presented in a centralized fashion, an equivalent decentralized execution scheme exists using message passing between neighboring agents since all information used is local. The mechanism maintains the locality constraints and is based only on the assumptions that each agent is able to communicate freely with another agent within a fixed radius.

\begin{algorithm}
\caption{Collision Avoidance for Continuous Space DMAR}\label{alg:collision:avoid}

\begin{algorithmic}
    \Repeat{}

    \State{1. Locate Waiting Agents.}
    
    \State{2. Detect Collisions.}
        
    \Repeat{ }
        
        \State{%
        3. Define Precedence.}
        
        \State{4. Backtrack or Move ahead.}
        
        \State{5. Decide how long to Wait.}

    \Until{there are no collisions}
    
    \State{6. Move Exploring Agents.}
    \Until{all agents have executed their control sequences} 
    \State{}
    \State{\textbf{Step Descriptions:}}
    \State{\textbf{Locate Waiting Agents:} Locate and note positions of agents that are exploring but have found a near-by task and are therefore waiting.}
    \State{\textbf{Detect Collisions:} Execute a single move for each agent that is in a cluster and has remaining moves to execute. After performing a single move for each such agent, there may exist locations where there are two or more agents trying to occupy the same position (collisions).}
    \State{\textbf{Define Precedence:} For each collision, we define a priority order and decide which agent continues to the new location  $\ell$ and which agent must remain at their previous location. Agents for whom $\ell$ is the final position (i.e. no additional moves remain in the control sequence or only \textsc{wait} moves remain), or those agents that are already waiting at $\ell$ are given preference.}
    \State{\textbf{Backtrack or Move ahead:} Agents with remaining moves are considered for movement away from the collision location by consulting their respective control sequences. This is done iteratively by finding next position for the agent (a trajectory is created one step at a time according to computed controls) until a location which is collision-free(safe) is found. If a safe location is found then the agent directly moves to that location, circumventing any resulting locations in between. If a safe location is not found, then the agent will simply wait at its original position for a specified amount of time.}
    \State{\textbf{Decide how long to Wait:} Now there are two scenarios to consider: (1.) a safe location may open in future for the agent that has backtracked, or (2.) there is no chance of any safe location in the trajectory in the future. In scenario (1.), the agent simply backtracks with a wait move and will try to execute its move again before the end of the current DMAR round. In scenario (2.), the agent will backtrack and call that location its final position. These scenarios can be determined by each agent from examination of the complete set of control sequences given.}
    \State{\textbf{Move Exploring Agents:} Exploring agents take lowest precedent. If an exploring agent would make a random move that results in a collision, the agent backtracks and chooses another random move. The agent will simply wait if all potential moves result in a collision.}

\end{algorithmic}
\end{algorithm}

\section{Compute specifications}
We used two sources of compute for our experiments:
\begin{enumerate}
    \item Traditional x86 compute nodes were used which contain two Intel Xeon E5-2680 v4 CPUs running at 2.40GHz with at least 128 GB of RAM. There are 28 high speed Broadwell class CPU cores per node. The operating system is CentOS Linux 7 (Core), Kernel: Linux 3.10.0-693.21.1.el7.x86\_64, Architecture: x86-64.
    \item A single compute node with AMD Ryzen 9 7950X Processor with 16 cores and 32 threads along with 64 GB of RAM. Physics based countinuous space simulations were exclusively run on this node. The operating system is Ubuntu 22.04.2 LTS, Kernel: Linux 6.2.0-76060200-generic, Architecture: x86-64.
\end{enumerate}

\section{Source Code}
Complete source code for our simulator can be found at \\
\href{https://github.com/SOPSLab/Distributed-Multi-Agent-Rollout}{https://github.com/SOPSLab/Distributed-Multi-Agent-Rollout}. The raw data from our simulations is available upon request.

\end{document}